\documentclass[letterpaper,journal]{IEEEtran}
\usepackage{amsmath}
\usepackage{algorithmic}
\usepackage{algorithm}
\usepackage{subfigure}
\usepackage{stfloats}
\usepackage{url}
\usepackage{graphicx}
\usepackage{cite}
\usepackage{siunitx}
\usepackage{booktabs} 
\usepackage{algorithm}
\usepackage{algorithmic}
\usepackage{amsthm}
\usepackage{microtype} 
\usepackage{balance}
\usepackage[hidelinks,implicit=false]{hyperref}
\usepackage{xcolor}
\usepackage{amssymb}
\usepackage{amsfonts}

\usepackage{xspace}
\newcommand{\safemath}[2]{\newcommand{#1}{\ensuremath{#2}\xspace}}

\usepackage{fancyref}

\newcommand{\conditionaltextstyle}{}
\safemath{\bma}{\mathbf{a}}
\safemath{\bmb}{\mathbf{b}}
\safemath{\bme}{\mathbf{e}}
\safemath{\bmn}{\mathbf{n}}
\safemath{\bmr}{\mathbf{r}}
\safemath{\bms}{\mathbf{s}}
\safemath{\bmx}{\mathbf{x}}
\safemath{\bmy}{\mathbf{y}}
\safemath{\bmz}{\mathbf{z}}
\safemath{\bF}{\mathbf{F}}
\safemath{\bH}{\mathbf{H}}
\safemath{\bI}{\mathbf{I}}
\safemath{\bN}{\mathbf{N}}
\safemath{\bQ}{\mathbf{Q}}
\safemath{\setC}{\mathcal{C}}
\safemath{\setN}{\mathcal{N}}
\safemath{\setO}{\mathcal{O}}
\safemath{\bZero}{\mathbf{0}}
\safemath{\opE}{\mathbb{E}}
\safemath{\complexset}{\mathbb{C}}
\safemath{\reals}{\mathbb{R}}
\safemath{\define}{\triangleq}			
\newcommand{\Ex}[2]{\ensuremath{\Exop_{#1}\lefto[#2\right]}} 	

\safemath{\SNR}{\textit{SNR}} 				
\safemath{\No}{N_0}							
\safemath{\Es}{E_s}		
\newcommand*{\fancyrefthmlabelprefix}{thm}		
\newcommand*{\fancyrefcorlabelprefix}{cor}		
\newcommand*{\fancyrefremlabelprefix}{rem}		
\newcommand*{\fancyreflemlabelprefix}{lem}		
\newcommand*{\fancyrefapplabelprefix}{app}		
\newcommand*{\fancyrefdeflabelprefix}{def}		
\newcommand*{\fancyrefalglabelprefix}{alg}		
\newcommand*{\fancyreftbllabelprefix}{tbl}		
\newcommand*{\fancyreftestlabelprefix}{est}		
\newcommand*{\fancyrefsyslabelprefix}{sys}		
\frefformat{vario}{\fancyrefseclabelprefix}{Section~#1}
\frefformat{vario}{\fancyreffiglabelprefix}{Figure~#1}
\frefformat{vario}{\fancyrefthmlabelprefix}{Theorem~#1}
\frefformat{vario}{\fancyrefcorlabelprefix}{Corollary~#1}
\frefformat{vario}{\fancyrefremlabelprefix}{Remark~#1}
\frefformat{vario}{\fancyreflemlabelprefix}{Lemma~#1}
\frefformat{vario}{\fancyrefapplabelprefix}{Appendix~#1}
\frefformat{vario}{\fancyrefdeflabelprefix}{Definition~#1}
\frefformat{vario}{\fancyrefalglabelprefix}{Algorithm~#1}
\frefformat{vario}{\fancyreftbllabelprefix}{Table~#1}
\frefformat{vario}{\fancyreftestlabelprefix}{Estimator~#1}
\frefformat{vario}{\fancyrefsyslabelprefix}{System Model~#1}
\frefformat{vario}{\fancyrefeqlabelprefix}{(#1)}
\newtheorem{theorem}{Theorem}
\newtheorem{corollary}{Corollary}
\newtheorem{remark}{Remark}
\newtheorem{lemma}{Lemma}
\newtheorem{definition}{Definition}
\newtheorem{estimator}{Estimator}
\newtheorem{system}{System Model}
\renewcommand{\algorithmiccomment}[1]{\bgroup\hfill\scriptsize{//~#1}\egroup}
\DeclareMathOperator{\sign}{sign}			
\newcommand{\abssquared}[1]{|#1|^2} 
\safemath{\Herm}{\textnormal{H}}
\safemath{\CN}{\mathcal{CN}}
\newcommand{\festnoargs}[0]{\eta} 
\newcommand{\fest}[1]{\festnoargs(#1)} 
\newcommand{\median}[0]{\mathsf{m}} 
\newcommand{\samplemedian}[0]{\overline{\median}}
\newcommand{\zvariable}{\bmz}
\newcommand{\zvariablee}{z}
\newcommand{\SURE}[0]{\textit{SURE}}
\newcommand{\Eo}[0]{\textit{MSE}}
\newcommand{\MSE}[0]{\Eo}
\newcommand{\Eh}[0]{E_s/p}
\newcommand{\sampleNo}[0]{\overline{N}_0}
\newcommand{\sampleEs}[0]{\overline{E}_s}
\newcommand{\sampleSNR}[0]{\overline{\SNR}}
\newcommand{\sampleEo}[0]{\overline{\Eo}}
\newcommand{\estimatedNo}[0]{\widehat{N}_0}
\newcommand{\estimatedEs}[0]{\widehat{E}_s}
\newcommand{\estimatedSNR}[0]{\widehat{\SNR}}
\newcommand{\estimatedMSE}[0]{\widehat{\Eo}}
\newcommand{\estimatedp}[0]{\hat p}
\newcommand{\estimatedNosandwich}[0]{\widehat{N}_0(\estimatedp)}
\newcommand{\estimatedNolargeD}[0]{\lim_{D\to\infty}\estimatedNo}
\newcommand{\estimatedNoEM}[0]{\widehat{N}_0^\text{EM}}
\newcommand{\estimatedEsEM}[0]{\widehat{E}_s^\text{EM}}
\newcommand{\estimatedpEM}[0]{\hat p^\text{EM}}
\newcommand{\plimNo}[0]{\estimatedNo\xrightarrow[\protect{\raisebox{3pt}[0pt][0pt]{\ensuremath{ \scriptstyle{D\to\infty}}}}]{\protect{\raisebox{-0.5pt}[0pt][0pt]{\ensuremath{ \scriptstyle{prob.}}}}}}
\newcommand{\Noinit}{\No^\text{init}}
\newcommand{\pinit}{p^\text{init}}
\newcommand{\nitmax}{K^\text{max}}
\newcommand{\tol}{\xi}
\newcommand{\nit}{K}
\newcommand{\pmax}[0]{p^\text{max}}
\newcommand{\vaold}{v_a^\text{old}}
\newcommand{\vbold}{v_b^\text{old}}
\newcommand{\nitidx}{k}
\newcommand{\LB}[0]{\textit{LB}}
\newcommand{\UB}[0]{\textit{UB}}
\newcommand{\zeronorm}{L}
\newcommand{\snonzeros}{\bms^\text{nz}}
\newcommand{\snonzerose}[1]{s^\text{nz}_{#1}}
\newcommand{\U}[0]{U} 
\newcommand{\B}[0]{D} 
\newcommand{\HML}{\bH^\text{ML}}
\newcommand{\NCE}{\bN^\text{CE}}
\newcommand{\NoCE}{\No^\text{CE}}
\newcommand{\Hbeamspace}{\widetilde{\bH}}
\newcommand{\HMLbeamspace}{\widetilde{\bH}^\text{ML}}
\newcommand{\NCEbeamspace}{\widetilde{\bN}^\text{CE}}
\newcommand{\HMLonebit}{\bH^\text{1-bit ML}}
\newcommand{\HMLonebitbeamspace}{\widetilde{\bH}^\text{1-bit ML}}

\begin{document}

\title{Low-Complexity Blind Parameter Estimation\\ in Wireless Systems with Noisy Sparse Signals}

\author{Alexandra Gallyas-Sanhueza  and Christoph Studer
\thanks{A.~Gallyas-Sanhueza is with the School of Electrical and Computer Engineering, Cornell University, Ithaca, NY; email: ag753@cornell.edu}
\thanks{C.~Studer is with the Department of Information Technology and Electrical Engineering, ETH Zurich, Zurich, Switzerland; email: studer@ethz.ch}
\thanks{The work of AGS and CS was supported in part by ComSenTer, one of six centers in JUMP, a Semiconductor Research Corporation (SRC) program sponsored by DARPA, and in part by the US National Science Foundation (NSF) under grants CNS-1717559 and ECCS-1824379.}
\thanks{Part of this work was presented at the IEEE International Conference on Communications (ICC) 2021 \cite{gallyas21a}. 
This journal paper extends our work by (i) including a novel parametric noise power estimator with improved accuracy, (ii) evaluating the proposed blind estimators as an initializer for an expectation-maximization algorithm, and (iii) {adding} two applications examples.}
\thanks{MATLAB code to reproduce our simulations {is} available on GitHub:  {\mbox{\protect\url{https://github.com/IIP-Group/blind_and_nonparametric_estimators}}}.}
\thanks{The authors thank Arian Maleki, Ramina Ghods, Charles Jeon, and Seyed Hadi Mirfarshbafan for discussions on signal recovery 
using SURE, and Haochuan Song for sharing the cell-free system simulator from \cite{song20a}.}
}

\maketitle

\begin{abstract}
Baseband processing algorithms often require knowledge of the noise power, signal power, or signal-to-noise ratio (SNR). 
In practice, these parameters are typically unknown and must be estimated.
Furthermore, the mean-square error (MSE) is a desirable metric to be minimized in a variety of estimation and signal recovery algorithms.
However, the MSE cannot directly be used as it depends on the true signal that is generally unknown to the estimator.
In this paper, we propose novel \emph{blind} estimators for the average noise power, average receive signal power, SNR, and MSE.
The proposed estimators can be computed at low complexity and solely rely on the large-dimensional and sparse nature of the processed data. 
Our estimators can be used (i) to quickly track some of the key system parameters while avoiding additional pilot overhead, (ii) to design low-complexity nonparametric algorithms that require such quantities, and (iii) to accelerate more sophisticated estimation or recovery algorithms.
We conduct a theoretical analysis of the proposed estimators for a Bernoulli complex Gaussian (BCG) prior, and we demonstrate their efficacy via synthetic experiments. 
We also provide three application examples that deviate from the BCG prior in millimeter-wave multi-antenna and cell-free wireless systems for which we develop nonparametric denoising algorithms that improve channel-estimation accuracy with a performance comparable to denoisers that assume perfect knowledge of the system parameters.
\end{abstract}


\section{Introduction}
\IEEEPARstart{A}{ccurate} 
 knowledge of system parameters, such as the average noise power, average signal power, and/or signal-to-noise ratio (SNR), is critical in wireless communication systems, as many baseband processing tasks rely on these quantities~\cite{schenk2008rf}.
Virtually all existing wireless systems dedicate training phases to estimate such parameters. 
These training phases  typically  consist of sending \emph{pilots}: signals that are known to the receiver and enable estimation of the desired parameters.
As pilots do not convey information, minimizing the pilot overhead is desirable in practice.
Furthermore, parameter estimation in wireless systems operating at millimeter-wave (mmWave) frequencies must be done frequently, since the propagation conditions can change at fast rates, e.g.,  blockers or interferers may appear or disappear quickly~\cite{rappaport13a}. Thus, it is even more important to reduce the pilot overhead.
In addition, such systems are expected to support several GHz of bandwidth and basestations will consist of a large number of antenna elements. It is  therefore important to develop \emph{low-complexity} solutions that quickly and accurately track such parameters for high-dimensional problems that must be processed at fast rates. 

From a parameter estimation perspective, it is beneficial that many modern wireless communication systems often deal with high-dimensional data. For example, all-digital massive multiple-input multiple-output (MIMO) basestations are expected to be equipped with hundreds of antennas~\cite{rusek14a} or orthogonal frequency-division multiplexing (OFDM) systems will support thousands of subcarriers \cite{3gpp20a}. 
Since many of these high-dimensional signals arising in such systems exhibit structure (e.g., are sparse or are taken from a discrete set), one can design statistical methods that blindly estimate critical parameters without requiring a dedicated training phase. 

In this paper, we focus on noisy observations of signal vectors that are \emph{sparse}, i.e., only few entries carry most of the signals' energy.
Examples of sparse vectors in wireless systems include
(i) the beamspace-domain representation of all-digital mmWave multi-antenna channel vectors~\cite{ghods19a,mirfarshbafan2019beamspace,gallyas20a},
(ii) the delay-domain representation of OFDM channel vectors~\cite{upadhya2016risk}, and 
(iii) the antenna-domain representation of channel vectors in cell-free MIMO wireless systems~\cite{gholamipourfard21a}.
We will explain how sparsity can be exploited to estimate parameters and denoise noisy observations of sparse vectors. 
In Sections \ref{sec:nonparametricestimators}, \ref{sec:theory}, and \ref{sec:synthetic_results}, we decouple our results from wireless communication applications and study the general setting. 
In \fref{sec:channel_denoising}, we apply our estimators and algorithms to three distinct applications in wireless systems.

In what follows, we will use the term ``\emph{blind}'' for estimators that do {not} use any pilots or training sequences and instead rely only on the signal statistics; blind estimators may have tuning parameters.
We will use the term ``\emph{nonparametric}'' for estimators that do not need knowledge of system parameters and do not have parameters that need to be tuned manually; nonparametric estimators may use pilots or training sequences.

\subsection{Prior Art in Blind and Nonparametric Estimation}
Many of the existing blind noise power and SNR estimators exploit modulation-specific structure, such as the cyclic prefix redundancy in OFDM~\cite{socheleau09a,wang10b}, or the periodicity of synchronization sequences \cite{zivkovic09a}.
Expectation-maximization (EM) has also been used for blind noise power or SNR estimation~\cite{Das12a}, and for joint sparse signal recovery and noise power estimation~\cite{huang2020a,wipf04a}.
However, the iterative nature of Bayesian algorithms and EM, and their relatively high per-iteration complexity renders such methods unsuitable for real-time estimation in wireless systems that operate with high-dimensional data at gigabit-per-second sampling rates.
In contrast, we propose low-complexity blind estimators whose complexity only scales with $\setO(D)$, where~$D$ is the dimension of the processed data.
Our proposed low-complexity estimators can also be used as an initialization point to accelerate the convergence of existing EM algorithms.

Joint noise power estimation and sparse signal recovery was investigated  in \cite{chretien14a}; these methods require the choice of algorithm parameters, which affect the estimation accuracy and robustness.
A parameter-free version of sparse signal recovery that combines approximate message passing (AMP)~\cite{donoho2009message,rangan11a} with Stein's unbiased risk estimate (SURE)~\cite{stein81a,tibshirani2015stein} was proposed in \cite{mousavi17a}. 
Similarly, the nonparametric equalizer (NOPE)~\cite{ghods2017optimally} combines AMP with SURE to perform linear minimum mean-square error (MSE) equalization in massive MIMO systems without knowledge of the SNR.
A drawback of such algorithms is the high per-iteration  complexity, which prevents their use in wireless systems supporting large bandwidths and high-dimensional problems (see, e.g.,~\cite{maechler12a,bai12a} for hardware results of sparse signal recovery). We therefore focus on low complexity, blind, and nonparametric algorithms for the fully-determined setting (in contrast to compressive sensing where one has fewer measurements than unknowns), which finds use in many practical situations.
For example, all-digital massive MIMO architectures (which can be as energy efficient as hybrid analog-digital architectures~\cite{yan2019performance,roth18a,skrimponis20a}) and cell-free wireless systems can provide measurement vectors of the same dimension as the sparse signal. In OFDM systems, even though pilots are typically transmitted only on a subset of all subcarriers, interpolation and extrapolation algorithms can be used to extract channel state information on all subcarriers \cite{haene07a}; this also leads to the fully-determined setting that enables the use of our methods.

In this low-complexity setting, the concept of estimating tuning parameters directly from the noisy observations has been used recently for adaptive denoising of mmWave~\cite{ghods19a,mirfarshbafan2019beamspace,gallyas20a} or OFDM~\cite{upadhya2016risk} channel vectors.
Such denoising algorithms typically require a tuning parameter: the denoising threshold. While SURE can be used to automatically determine the MSE-optimal denoising threshold, it still requires knowledge of the noise power.
In contrast to such results, we propose low-complexity blind estimators, which enable the design of nonparametric (i.e., parameter free) channel-vector denoising algorithms that deliver comparable performance to methods that assume perfect knowledge of the required parameters (e.g., the noise power).

Blind nonparametric algorithms have been proposed for denoising {of} real-valued signals. The authors in~\cite{donoho94} have used power estimation methods based on the median absolute deviation (MAD) of real-valued signals for wavelet denoising. 
The Python wavelet toolbox PyYAWT~\cite{pyyawt} includes MAD-based power estimation and adaptive wavelet denoising using SURE for real-valued signals. 
Our methods also build upon MAD and SURE, but  are suitable for complex-valued signals. In addition, we provide a detailed derivation and a theoretical analysis, and extend the general concept to estimate other quantities that frequently arise in wireless systems.
While some papers apply \emph{real-valued} MAD for noise power estimation in the complex-valued setting (see, e.g., \cite{alexander2000a} for magnetic resonance imaging),  there are non-negligible differences to the complex case. We therefore derive the complex-valued version, provide a theoretical accuracy analysis with a Bernoulli complex Gaussian (BCG) prior, and show application examples that deviate from this prior in order to highlight robustness and usefulness of our results.

\subsection{Contributions}
A variety of applications in communication systems deal with sparse and complex-valued signals whose observations are contaminated with noise.
For such a model, we propose novel low-complexity blind estimators for the average noise power, average signal power, and SNR. 
In addition, we propose a blind estimator for the MSE of an estimation function that aims to recover the sparse signal. We use this blind MSE estimate to design a novel nonparametric channel-vector denoising algorithm. 
We conduct a theoretical analysis of our estimators for a BCG prior, and we showcase simulation results with synthetic data in order to demonstrate the efficacy and limits of our estimators in finite dimensions.
In order to demonstrate the efficacy of our results in situations that deviate from a BCG prior, we provide three application examples of channel-vector denoising in mmWave and cell-free communication systems.
We also show that our low-complexity estimators can be used to accelerate the convergence (and, hence, reduce the complexity) of existing estimators with a concrete example of an EM-based algorithm.

\subsection{Notation}
Lowercase and uppercase boldface letters denote column vectors and matrices, respectively.
The $d$th entry of the vector~$\bma\in\complexset^D$ is~$a_d$; the real and imaginary parts are $\Re\{\bma\}$ and $\Im\{\bma\}$, respectively.
We use $\bmb \define \abssquared{\bma}$ to refer to $b_d=|a_d|^2$ for $d=1,\ldots,D$. For $\bma\in\complexset^D$, the vector 
$q$-norm is defined as $\|\bma\|_q \define \left(\sum_{d=1}^{D}{|a_d|^q}\right)^{1/q}$ for $q\geq1$ with $\|a\|_\infty \define \max_{d=1,\ldots,D} |a_d|$ and the $\ell_0$-pseudo-norm  $\|\bma\|_0$ counts the number of nonzero entries in $\bma$.
The identity matrix is~$\bI$ and the all-zeros vector is $\bZero$.
The discrete Fourier transform matrix is denoted by~$\bF$ and satisfies $\bF^\Herm\bF=\bI$, where the superscript $^\Herm$ denotes the  Hermitian (conjugate transpose) matrix.
An i.i.d. circularly-symmetric complex Gaussian random vector $\bmx\in\complexset^D$ with variance $E_x$ per complex dimension is denoted by $\bmx\sim\setC\setN(\bZero, E_x \bI)$ and its probability density function (PDF) evaluated at $\bmx$ is $f^\CN(\bmx; \bZero,E_x\bI)$.
Sample estimates are denoted by a bar, e.g., the sample variance $\overline{E}_x \define \frac{1}{D}\|\bmx\|_2^2$ of the random vector $\bmx\in\complexset^D$; 
statistical quantities are denoted by plain symbols, e.g., the variance $E_x \define \frac{1}{D}\Ex{}{\|\bmx\|_2^2}$, where $\Ex{}{\cdot}$ denotes expectation; 
blind estimators are denoted by a hat, e.g.,~$\widehat{E}_x$.
For $x\in\reals$, rounding towards plus and minus infinity is denoted by $\lceil x\rceil$ and $\lfloor x \rfloor$, respectively, and $[x]_+ \define \max\{x,0\}$. 
Convergence in probability of a random sequence $A_n$ to a random variable $A$ is $A_n\xrightarrow[\protect{\raisebox{3pt}[0pt][0pt]{\ensuremath{ \scriptstyle{n\to\infty}}}}]{\protect{\raisebox{-0.5pt}[0pt][0pt]{\ensuremath{ \scriptstyle{prob.}}}}} A$ and almost sure convergence is $A_n \xrightarrow[\protect{\raisebox{3pt}[0pt][0pt]{\ensuremath{ \scriptstyle{n\to\infty}}}}]{\protect{\raisebox{-0.5pt}[0pt][0pt]{\ensuremath{ \scriptstyle{a.s.}}}}} A$.

\section{Practical Guide to Low-Complexity Blind Estimators}
\label{sec:nonparametricestimators}

We now introduce two system models and propose low-complexity blind estimators for  the average noise and signal powers, SNR, and MSE. 
The derivation of the proposed estimators and an analysis of the key properties are provided in \fref{sec:theory}.

\subsection{System Models}
\label{sec:systemmodel}
We say that a complex-valued vector 
$\bms\in\complexset^D$ is \emph{sparse} if the number of nonzero entries is smaller than the dimension~$D$. As a sparsity measure, one can use, for example, the $\ell_0$-pseudo-norm~$\|\bms\|_0$. 
This definition of sparsity allows us to derive theoretical results, but in practice, our algorithms also work for \emph{approximately} sparse signals in which most entries are small compared to the noise (but not necessarily zero).
We will focus on the following two system models.

\begin{system} \label{sys:systemmodel1}
Let $\bms\in\complexset^D$ be a sparse signal with average power $\Es \define \frac{1}{D}\Ex{}{\|\bms\|_2^2}$.
We model the input-output relation of a noisy observation of the sparse signal as
\begin{align} \label{eq:inputoutputrelation}
\bmy = \bms + \bmn,
\end{align}
where $\bmy\in\complexset^D$ is the noisy observation and $\bmn\in\complexset^D$ models noise with $\bmn\sim\mathcal{CN}(\bZero,\No\bI)$.
We assume that the sparse signal vector $\bms$ and noise vector $\bmn$ are statistically independent. 
\end{system}
\fref{sys:systemmodel1} finds numerous applications in wireless communication systems. Prime examples are in describing estimated channel vectors (i) in multi-antenna mmWave systems, where the beamspace-domain representation of the channel vectors is typically sparse~\cite{ghods19a,mirfarshbafan2019beamspace,gallyas20a}, (ii) in OFDM systems, where the delay-domain representation of the channel vectors is typically sparse~\cite{upadhya2016risk}, or
(iii) in cell-free communication systems with centralized processing, where the antenna-domain representation of the channel vectors is typically sparse~\cite{gholamipourfard21a}. 
In what follows, we assume the sparse vector $\bms$ is unknown (in contrast to pilot-based estimation), which makes parameter estimation nontrivial in this blind scenario.

\begin{system} \label{sys:systemmodel2}  
Let $\bmy\in\complexset^D$ be a noisy observation as in~\fref{sys:systemmodel1}. 
Fix a weakly differentiable function\footnote{A weakly differentiable function may be nondifferentiable only in zero-measure sets (e.g., for particular values), and has to be differentiable everywhere~else.} $\protect{\festnoargs: \complexset \to \complexset}$ that operates entry-wise on vectors. 
We model the output after applying this function to the noisy observation as
\begin{align} \label{eq:inputoutputrelation2}
\fest{\bmy} = \bms + \bme,
\end{align}
where $\bme\in\mathbb{C}^D$ contains (likely non-Gaussian) residual distortion. We emphasize that the sparse signal vector $\bms$ and the residual distortion vector $\bme$ are not necessarily statistically independent. 
\end{system}

\fref{sys:systemmodel2} is relevant in the following scenarios:
(i) Estimating a sparse signal~$\bms$ from a noisy observation~$\bmy$ by applying an entry-wise denoising or estimation function, producing the signal estimate $\hat\bms \define \fest{\bmy}$; this scenario finds use for channel-vector denoising~\cite{ghods19a,mirfarshbafan2019beamspace,gallyas20a}. 
(ii) Modeling nonlinearities caused by hardware impairments~\cite{jacobsson18d}, in which case the distorted version of the noisy received signal can be expressed as $\bmr \define \fest{\bmy}$; this scenario finds use in signals sampled with low-resolution data converters~\cite{li17b,jacobsson17b}, for example.

\subsection{Low-Complexity Blind Nonparametric Estimators}
\label{sec:estimatorsusbection}

In what follows, we make use of the sample median, which we define as follows. 

\begin{definition}[Sample Median]
Let $\zvariable\in\reals^D$ be a real-valued vector and $\zvariable^\text{sort}\in\reals^D$ be its sorted version (entries sorted in ascending order). Then, the \emph{sample median} is defined as 
\begin{align} \label{eq:samplemedian}
\samplemedian(\zvariable) \define \conditionaltextstyle \frac{1}{2}\Big(\zvariablee^\text{sort}_{\lfloor(D+1)/2\rfloor}+\zvariablee^\text{sort}_{\lceil(D+1)/2\rceil}\Big).
\end{align}
\end{definition}
The sample median is robust to outliers \cite{rousseeuw1993alternatives,huber2004robust}, which makes it amenable to \fref{sys:systemmodel1}, as the nonzero entries of the sparse vector $\bms$ can be considered to be outliers for the purpose of separating the sparse signal from noise. 
We emphasize that the sample median can be computed at a complexity of $\setO(D)$ average time using quickselect~\cite{tibshirani09a} or of $\setO(D)$ deterministic time using the MedianOfNinthers algorithm~\cite{alexandrescu17a}.

We now propose a range of low-complexity blind estimators (no pilots required) for complex-valued signals that require no parameters. 

\begin{estimator}[Average Noise Power] \label{est:noisevariance}
Consider~\fref{sys:systemmodel1}. 
We propose the following blind estimator 
\begin{align} \label{eq:noiseestimator}
\estimatedNo \define \conditionaltextstyle  \frac{\samplemedian(\abssquared{\bmy})}{\log(2)}
\end{align}
{for} the average noise power defined as 
$\No \define \conditionaltextstyle \Ex{}{\|\bmn\|_2^2}/D$.
\end{estimator}

\fref{est:noisevariance} is blind as it only requires  the absolute square entries of the noisy observation~$\bmy$ in~\fref{eq:inputoutputrelation}.
The estimate $\estimatedNo$ can be computed efficiently in $\setO(D)$ time, since the most complex operation is computing the median of a vector of dimension~$D$. 
\fref{est:noisevariance} exploits sparsity in the signal $\bms$, but is independent of the signal sparsity, the signal power, or the statistical sparsity model.
It is, however, important to understand that the accuracy of this estimator depends on all of these factors as it relies on the fact that the nonzero entries of the sparse vector~$\bms$ can be treated as outliers for the purpose of estimating the average noise power.
We note that this noise power estimator can be seen as a complex-valued and squared version\footnote{The squared median absolute deviation (MAD) estimator for real-valued signals provided in~\cite{huber2004robust} corresponds to $\samplemedian(|\bmy|)^2$ whereas we propose to use $\samplemedian(|\bmy|^2)$. While $\samplemedian(|\bmy|)^2 \leq \samplemedian(|\bmy|^2)$ if $D$ is even, both estimators coincide if $D$ is odd. What is more, our scaling factor $\log(2)\approx 0.6931$ differs considerably from the widely-used scaling factor of $(\Phi^{-1}(3/4))^2\approx (0.6745)^2$ for real-valued signals~\cite{donoho94}. We reiterate that the latter is derived for power estimation of \emph{real-valued} Gaussians using the MAD estimator, while in our derivation we consider the case of \emph{complex-valued} Gaussians.} of the median absolute deviation (MAD) estimator~\cite{rousseeuw1993alternatives,pham-gia_01a}, where we use the assumption that the noise in \fref{sys:systemmodel1} is zero mean. 
The intuition behind this estimator (and the $\log(2)$ factor) is the fact that the entries ${|n_d|^2}/{(\No/2)}$, $d=1,\ldots,D$ are  $\chi^2$ distributed with two degrees of freedom, which have a median of $2\log(2)$, and that the median of $\abssquared{\bmy}$ is not significantly ``contaminated'' by the sparse signal.
{\fref{est:noisevariance} is used in the estimators proposed next.}

\begin{estimator}[Average Signal Power] \label{est:signalpower}
Consider~\fref{sys:systemmodel1}. 
We propose the following blind estimator 
\begin{align} \label{eq:signalpowerestimator}
\estimatedEs \define \conditionaltextstyle  \left[ \frac{\|\bmy\|_2^2}{D} - \estimatedNo \right]_{\!+}
\end{align}
for the average signal power defined as
$\Es \define \conditionaltextstyle \Ex{}{\|\bms\|_2^2}/D$.
\end{estimator}

\fref{est:signalpower} is blind as it only requires the sample estimate of the receive power $\overline{E}_y \define \|\bmy\|_2^2/D$ and the blind noise estimate~$\estimatedNo$ from \fref{est:noisevariance}.
$\estimatedEs$ can be computed efficiently in $\setO(D)$ time, since the most complex operation is computing~$\estimatedNo$. 
The intuition behind this estimator comes from subtracting the estimated noise power from the total receive power, as done previously in~\cite{wang10b} for an OFDM-specific estimator.

\begin{estimator}[Signal-to-Noise Ratio]  \label{est:SNR}
Consider~\fref{sys:systemmodel1}. 
We propose the following blind estimator 
\begin{align} \label{eq:SNRestimateomg}
\estimatedSNR \define \conditionaltextstyle \left[ \frac{\|\bmy\|_2^2}{D \estimatedNo} - 1 \right]_{\!+}
\end{align}
for the SNR defined as $\SNR \define \conditionaltextstyle {\Ex{}{\|\bms\|_2^2}}/{\Ex{}{\|\bmn\|_2^2}}$.
\end{estimator}

\fref{est:SNR} is blind as it only requires the sample estimate of the receive power $\overline{E}_y \define \|\bmy\|_2^2/D$ and the blind estimate~$\estimatedNo$ from \fref{est:noisevariance}.
$\estimatedSNR$ can also be computed efficiently in $\setO(D)$ time.
The intuition behind this estimator comes from dividing the estimated signal and noise powers, as done previously in~\cite{wang10b} for an OFDM-specific estimator.

\begin{estimator}[Mean-Square Error]   \label{est:MSE}
Consider~\fref{sys:systemmodel2} with a fixed function $\festnoargs: \complexset \to \complexset$.
We propose the following blind estimator  
\begin{align} 
\estimatedMSE \define\, & \conditionaltextstyle \frac{1}{D}\|\fest{\bmy}-\bmy\|_2^2 - \estimatedNo \notag \\
 &  + \conditionaltextstyle \frac{\estimatedNo}{D}\sum_{d=1}^D\left(\frac{\partial\Re\{\fest{y_d}\}}{\partial\Re\{y_d\}}+\frac{\partial\Im\{\fest{y_d}\}}{\partial\Im\{y_d\}}\right) \label{eq:MSEnonparametricexplicitform}
\end{align}
for the MSE defined as $\MSE \define \Ex{}{\|\fest{\bmy}-\bms\|_2^2}/D = \Ex{}{\|\bme\|_2^2}/D$.
\end{estimator}

\fref{est:MSE} is blind as it only requires the receive signal $\bmy$, the blind estimate~$\estimatedNo$ from \fref{est:noisevariance}, and the function $\festnoargs$. 
The complexity of the proposed MSE estimator depends on the function $\festnoargs$. For example, \fref{eq:MSEnonparametricexplicitform} can be computed efficiently in $\setO(D)$ time  for the soft-thresholding function with a given threshold. Even if the threshold is not given, searching for the best threshold and applying the soft-thresholding function can be done in $\setO(D\log(D))$ time using the methods developed for the BEACHES algorithm in \cite{ghods19a}.
The MSE is a frequently used metric to evaluate the performance of estimation algorithms. 
Our blind MSE estimate, since it is independent of~$\bms$, can be used to automatically tune parameters in estimators.
The intuition behind this estimator relies on SURE, and we refer the interested reader to \cite{tibshirani2015stein} for an accessible derivation in the real-valued case and to \cite{ghods19a,mirfarshbafan2019beamspace} for a derivation in the complex-valued case.
\fref{est:MSE} is used to obtain the nonparametric channel-vector denoising algorithm described in \fref{sec:channel_denoising}.

\subsection{Low-Complexity Blind Parametric Estimators}

We now propose a low-complexity blind estimator (no pilots or training sequences required) that takes an estimate~$\estimatedp$ of the activity rate as  a parameter. We then propose a family of parametric estimators for the activity rate.
	
\begin{estimator}[Average Noise Power with Estimated SNR and Activity Rate Corrections] \label{est:sandwich}
Consider \fref{sys:systemmodel1}, the low-complexity blind estimates $\estimatedSNR$ from \fref{est:SNR}, and a parameter~$\estimatedp$ that is an estimate of the activity rate $p$.
We propose the following blind parametric estimator
\begin{align} 
\estimatedNosandwich \define\, \conditionaltextstyle \frac{1}{2}\estimatedNo & \left({ \max\left\{ \frac{\log(2)}{\log\left(\frac{2-2\estimatedp}{1-2\estimatedp}\right)}, \frac{1}{1+\estimatedSNR} \right\} }\right. \notag \\
& \quad \left. \vphantom{\left\{ \frac{lala}{\left(\frac{lala\estimatedp}{lala\estimatedp}\right)}\right\} } +  (1-\estimatedp)+\frac{\estimatedp^2}{\estimatedp+\estimatedSNR}  \right) \label{eq:sandwich_estimator}
\end{align}
{for} the average noise power $\No$. 
\end{estimator}
\fref{est:sandwich} is blind as it only requires the blind estimates $\estimatedNo$, $\estimatedSNR$, but is parametric as it depends on the activity rate estimate $\estimatedp$.
$\estimatedNosandwich$ can be computed efficiently in $\setO(D)$ time, since the most complex operations are computing $\estimatedNo$ and $\estimatedSNR$, and eventually $\estimatedp$ (but here we consider~$\estimatedp$ as a given parameter and ignore the complexity associated with obtaining~it).
The intuition behind this estimator will become clear after we present \fref{thm:mainresult}, as it is derived from averaging a lower and an upper bound on $\No$.
As shown in \fref{sec:synthetic_results}, this parametric estimate $\estimatedNosandwich$ often yields better accuracy than the nonparametric estimate~$\estimatedNo$. 

Since in some applications an estimate for $\estimatedp$ may be unavailable, we next propose a family of estimators that attempt to extract the activity rate~$p$ directly from the noisy observation vector~$\bmy$. Such estimators can, for example, be used to substitute $\estimatedp$ in \fref{eq:sandwich_estimator}.

\begin{estimator}[Activity Rate] \label{est:p}
Consider \fref{sys:systemmodel1} and integers $1\leq q < r$. We propose the following family of blind parametric estimators\footnote{In practice, we use $\min\{0.499,\estimatedp(q,r)\}$ in place of $\estimatedp(q,r)$, so that $\log\!\Big(\frac{2-2\estimatedp(q,r)}{1-2\estimatedp(q,r)}\Big)$ is always well-defined, as required~by~\fref{eq:probabilitycondition}.}
\begin{align} \label{eq:est_p}
\estimatedp(q,r) \define \conditionaltextstyle \frac{1}{D}\left(\frac{\|\bmy\|_q}{\|\bmy\|_r}\right)^{\frac{1}{1/q-1/r}}
\end{align}
for the activity rate\footnote{We define the activity rate as the fraction of nonzero entries of the vector~$\bms$. Values of $p$ close to $0$ indicate the vector is sparse and $p=1$ indicates that all entries are nonzero.} defined as $p \define \Ex{}{\|\bms\|_0}/D$.
\end{estimator}

\fref{est:p} is blind as it only requires the receive vector $\bmy$, but is parametric as it requires a choice for $q$ and $r$. $\estimatedp(q,r)$ can be computed efficiently in $\setO(D)$ time and, among others, the following choices for $q$ and $r$ require low complexity: 
$\estimatedp{(1,2)} \define \frac{1}{D}\left(\frac{\|\bmy\|_1}{\|\bmy\|_2}\right)^{2}$, 
$\estimatedp{(1,\infty)} \define \frac{1}{D}\left(\frac{\|\bmy\|_1}{\|\bmy\|_\infty}\right)$, 
$\estimatedp{(2,4)} \define \frac{1}{D}\left(\frac{\|\bmy\|_2}{\|\bmy\|_4}\right)^{4}$, and 
$\estimatedp{(2,\infty)} \define \frac{1}{D}\left(\frac{\|\bmy\|_2}{\|\bmy\|_\infty}\right)^{2}$.
The parameters $q$ and $r$ must be chosen according to simulations, as we are unaware of a principled and reliable way to determine them. In our simulations, the choice~$\estimatedp(1,\infty)$ performed~best.

\subsection{Blind Parametric Estimator Based on Expectation-Maximization (EM)}
As a baseline, we {also} consider a blind EM estimator (no pilots required) that requires initialization values and algorithm parameters {that} determine the convergence criterion.  

\begin{estimator}[Noise Power, Signal Power, and Activity Rate]   \label{est:EM} 
Consider~\fref{sys:systemmodel1}. \fref{alg:EM}, initialized with $\protect{\Noinit<\|\bmy\|_2^2/D}$ and $\pinit<0.5$, simultaneously estimates the noise power $\No$, the signal power $\Es$, and the activity rate $p$.
\end{estimator}

\begin{algorithm*}
\caption{\strut Baseline EM (adapted from~\cite[Alg.\,8.1]{hastie01a} to circularly-symmetric complex Gaussians) 
\label{alg:EM}}
\begin{algorithmic}[1]
\STATE {\bf input} $\bmy$, $\nitmax$, $\tol$, $\Noinit$, and $\pinit$ \COMMENT{noisy signal, maximum iterations, tolerance, initial noise variance, initial activity rate}
\STATE {\bf initialize} \COMMENT{initialize weights ($w$) and variances ($v$) of distributions A and B as follows:} 
\STATE \quad $w_a \leftarrow 1-\pinit$, $v_a \leftarrow \Noinit$ \COMMENT{initialize A as noise}
\STATE \quad $w_b \leftarrow \pinit$, $v_b \leftarrow v_a+\frac{1}{\pinit}\left[\frac{\|y\|_2^2}{D}-\Noinit\right]_+$ \COMMENT{initialize B as signal plus noise, ensuring $w_a v_a + w_b v_b = {\|\bmy\|_2^2/D}$}
\STATE \quad $k \leftarrow 0$, $\vaold \leftarrow \infty$, $\vbold \leftarrow \infty$
\WHILE[iterate until convergence or maximum iterations reached\!\!\!]{$\tol < \frac{|v_a-\vaold|}{v_a} + \frac{|v_b-\vbold|}{v_b}$ \AND $\nitidx < \nitmax$} 
\STATE $\nitidx \leftarrow \nitidx+1$, $\vaold \leftarrow v_a$, $\vbold \leftarrow v_b$
\FOR{$d=1$ \TO $D$} 
\STATE $a_d \leftarrow \frac{w_a f^\CN(y_d;0,v_a)}{w_a f^\CN(y_d;0,v_a)+w_b f^\CN(y_d;0,v_b)}$ \COMMENT{likelihood that $y_d$ comes from distribution A}
\STATE $b_d \leftarrow \frac{w_b f^\CN(y_d;0,v_b)}{w_a f^\CN(y_d;0,v_a)+w_b f^\CN(y_d;0,v_b)}$ \COMMENT{likelihood that $y_d$ comes from distribution B}
\ENDFOR
\STATE $w_a \leftarrow \frac{1}{D}\sum_{d=1}^{D} a_d$, $w_b \leftarrow 1-w_a = \frac{1}{D}\sum_{d=1}^{D} b_d$ \COMMENT{update weights based on updated likelihoods} 
\STATE $v_a \leftarrow \frac{1}{w_a D}{\sum_{d=1}^{D}{a_d |y_d|^2}}$, $v_b \leftarrow \frac{1}{w_b D}{\sum_{d=1}^{D}{b_d |y_d|^2}}$ \COMMENT{update variances based on updated likelihoods}
\ENDWHILE																
\IF[assign the smallest variance to noise, and the largest variance to signal plus noise\!\!\!]{$v_a>v_b$}
\STATE $\estimatedNoEM \leftarrow v_b$, $\estimatedEsEM \leftarrow w_{a} (v_a-v_b)$,  $\estimatedpEM \leftarrow w_a$
\ELSE
\STATE $\estimatedNoEM \leftarrow v_a$, $\estimatedEsEM \leftarrow w_{b} (v_b-v_a)$,  $\estimatedpEM \leftarrow w_b$
\ENDIF
\STATE $\nit \leftarrow \nitidx$ \COMMENT{save number of iterations until convergence or stopping condition met}
\STATE  {\bf return} $\estimatedNoEM$, $\estimatedEsEM$, and $\estimatedpEM$ \COMMENT{estimated noise variance, estimated signal variance, estimated activity rate}
\end{algorithmic}
\end{algorithm*}

\fref{est:EM} is blind as it only requires the noisy observation~$\bmy$, but is parametric as it needs a choice for the maximum number of iterations $\nitmax$, the tolerance $\tol$, and initialization values for the noise power $\Noinit$ and activity rate $\pinit$. The total number of EM iterations $\nit$ is not fixed but depends on $\nitmax$, $\tol$, $\Noinit$, $\pinit$, and on the input $\bmy$ itself. 
The complexity of \fref{est:EM} is $\setO(KD)$. 
We note that this estimator is a variant of a classical EM algorithm for a two-component Gaussian mixture~\cite{hastie01a}, where we use the assumption that the signal and the noise in~\fref{sys:systemmodel1} are zero mean and complex valued. 
The intuition behind this estimator is the fact that each entry $|y_d|$, $d=1,\ldots,D$, of vector $\bmy$  contains either noise or signal-plus-noise, and those two cases have Gaussian distribution with different variances.

We note that this baseline EM algorithm is only a minor variation of the method in~\cite[Alg.~8.1]{hastie01a}. The iterative nature of such methods, however, results in (often significantly) higher complexity than our estimators.
With this in mind, we propose an improved version that we call ``accelerated EM,'' which simply consists of initializing the baseline EM algorithm using our blind nonparametric noise variance estimator. As we will see in \fref{sec:accelerated_convergence}, this accelerated EM variant drastically reduces the number of iterations needed for convergence without degrading accuracy.

{\subsection{Summary of Proposed Power Estimation and Denoising Algorithms}}

\begin{table*}
\caption{Complexity and Accuracy Summary. $D$ is the signal dimension and $K^\text{ACC} \ll K^\text{BL}$ refer to the number of iterations in the accelerated EM and baseline EM algorithms, respectively.
}
\label{tbl:comparison}
\centering
\begin{tabular}{@{}lcccc@{}}
\toprule
\multicolumn{1}{c}{\textbf{}} & \multicolumn{2}{c}{{Complexity}} & \multicolumn{2}{c}{{Accuracy}}\\
\multicolumn{1}{c}{\textbf{}} & \multicolumn{1}{c}{Power estimation} & \multicolumn{1}{c}{Denoising} & \multicolumn{1}{c}{Synthetic data} & \multicolumn{1}{c@{}}{Realistic channels}\\
\midrule
{Baseline EM} & $\setO(K^\text{BL}D)$ & $\setO(K^\text{BL}D+D\log(D))$ & (\checkmark\checkmark\checkmark)  & (\checkmark\checkmark)\\
{Accelerated EM} & $\setO(K^\text{ACC}D)$ & $\setO(K^\text{ACC}D+D\log(D))$ & (\checkmark\checkmark\checkmark)  & (\checkmark\checkmark)\\
{Nonparametric} &$\setO(D)$ & $\setO(D\log(D))$ & (\checkmark)  & (\checkmark\checkmark\checkmark)\\
{Parametric} &$\setO(D)$ & $\setO(D\log(D))$ & (\checkmark\checkmark)  & (\checkmark\checkmark\checkmark)\\
\bottomrule
\end{tabular}
\end{table*}

\fref{tbl:comparison} summarizes the complexity and accuracy of the different estimators. ``Baseline EM'' refers to \fref{est:EM}, ``accelerated EM'' to \fref{est:EM} initialized using \fref{est:noisevariance}, ``nonparametric'' to \fref{est:noisevariance}, and ``parametric'' to \fref{est:sandwich}. 
The complexity for blind noise power estimation is mentioned below the definition of each of these algorithms. The complexity for denoising is the complexity of estimating the noise power plus the complexity of the BEACHES algorithm from \cite{mirfarshbafan2019beamspace}. Since BEACHES already sorts the magnitudes of the noisy signal, the nonparametric and parametric estimators that use the median require no additional complexity for estimating the noise power.
Anticipating the results shown in \fref{sec:synthetic_results} and \fref{sec:channel_denoising}, we illustrate (qualitatively) the accuracy of the estimators with synthetic data that perfectly matches the BCG prior, and with practical examples that deviate from this prior. 



\section{Theory}
\label{sec:theory}

We first show that the sample median approaches the median for $D\to\infty$ and introduce our statistical model for sparse vectors. We then derive and analyze Estimators 1 to 7. 
The observations made in this section are valid in the large-dimension limit and for the noisy BCG model to be introduced in \fref{def:noisyBCG}. 
We use simulations to demonstrate the accuracy of our estimators for finite (and small) dimensions~$D$ with the noisy BCG model
in \fref{sec:synthetic_results}.
To demonstrate the efficacy of our methods in practical scenarios with signals that deviate from the BCG model, we evaluate our denoising algorithms in three distinct scenarios in \fref{sec:channel_denoising}.

\subsection{Convergence of the Sample Median for $D\to\infty$}
We will use the following definition of the median.
\begin{definition}[Median] \label{def:median}
Let $X$ be an absolutely continuous random variable (RV) with cumulative distribution function (CDF) $F_X(x)$.
Then, the \emph{median} $\median_X$ of $X$ is defined as
\begin{align}
F_X(\median_X) =  \conditionaltextstyle \frac{1}{2}. \label{eq:median}
\end{align}
\end{definition}

While, analogously to the central limit theorem, the sample median is approximately Gaussian if $D$ is large (see, e.g.,~\cite{MM08}), we will only use the following result. 

\begin{lemma}[Lem.~C.1 from \cite{MM08}] \label{lem:convergenceofmedian}
Let $X$ be a RV whose PDF is differentiable in some neighborhood of the median $\median_X$ and vector $\bmx$ contain i.i.d.\ samples of $X$.
Then, for any $c>0$ the sample median $\samplemedian(\bmx)$ satisfies  
\begin{align}
\lim_{D\to \infty} \Pr [ |\samplemedian(\bmx)-\median_X|\geq c ] = 0.
\end{align}
\end{lemma}

This result implies that in the \emph{large-dimension limit} ($\protect{D\to\infty}$), the sample median~$\samplemedian(\bmx)$ converges in probability to the median $\median_X$.
Hence, by observing a sufficiently large number of samples, which is possible in modern multi-antenna mmWave or OFDM systems, we can accurately estimate the median~$\median_X$. 

\subsection{Statistical Model for Complex-Valued Sparse Vectors}
\label{sec:statisticalmodel}
To derive and analyze the blind estimators proposed in \fref{sec:nonparametricestimators}, we need a statistical model for the sparse signal $\bms$.
This model should (i) have as few parameters as possible while being able to model a large class of complex-valued sparse vectors typically arising in communication systems and (ii) facilitate a theoretical analysis.
In what follows, we consider BCG random vectors \cite{vila11a,rangan11a}, which allow control over the signal sparsity and the signal power.
We reiterate that the BCG model is instrumental \emph{only} for our analysis. The provided simulation results in \fref{sec:channel_denoising} will show that the proposed estimators exhibit robustness to model mismatch, e.g., for signals that are not necessarily i.i.d. Gaussian or circularly~symmetric.

\begin{definition}[BCG Random Vector] \label{def:BCG}
A sparse vector $\bms\in\complexset^D$ is BCG if each entry is nonzero with probability $p\in(0,1]$, and the nonzero entries are i.i.d.\ circularly-symmetric complex Gaussian with variance $\Eh$.
The PDF of each entry $s_d$, $d=1,\ldots,D$, is therefore given by
\begin{align} \label{eq:BCG}
 f_S(s_d) \define \conditionaltextstyle (1-p) \delta(s_d) + p \frac{1}{\pi \Eh} e^{-\frac{|s_d|^2}{\Eh}},
\end{align}
where $\delta(\cdot)$ is the Dirac delta distribution. 
\end{definition}

With this model, the activity rate is $p = \Ex{}{\|\bms\|_0}/D$ (meaning the expected number of nonzero entries is $\protect{\Ex{}{\|\bms\|_0} = pD}$), and the average power of the sparse signal  vector $\bms$ is $\Es = \frac{1}{D}\Ex{}{\|\bms\|_2^2}$.

In \fref{sys:systemmodel1}, we assumed that the noise vector $\bmn$ is i.i.d.\ circularly-symmetric complex Gaussian with variance~$\No$ per complex entry. Hence, the PDF of each entry $n_d$, $d = 1,\ldots,D$, is given by $f^\CN(n_d;0,\No) \define \frac{1}{\pi\No}e^{-|n_d|^2/\No}$.
Consequently, 
if $\bms$ is a BCG random vector,
then the PDF of the noisy observation vector~$\bmy=\bms+\bmn$ is as follows.

\begin{definition}[Noisy BCG Random Vector] \label{def:noisyBCG}
The PDF of the entries $y_d$, $d=1,\ldots,D$, of a BCG random vector per \fref{def:BCG} observed as  in~\fref{sys:systemmodel1}  is given by
\begin{align}
\conditionaltextstyle f_Y(y_d) \define\ & (1-p) \frac{1}{\pi \No} e^{-\frac{|y_d|^2}{\No}}   \notag \\
& + \conditionaltextstyle p \frac{1}{\pi (\No+\Eh)} e^{-\frac{|y_d|^2}{\No+\Eh}}.
\end{align}
\end{definition}

For this signal and observation model, we are now able to derive and analyze Estimators 1~to~7.  
We will make frequent use of the entry-wise square of vector $\bmy$ that we will call $\bmz \define |\bmy|^2$.
We also define a random variable (RV) $Z$ with the same distribution as any of the i.i.d entries of $\bmz$, and let $\median_Z$ be the median of $Z$.

\subsection{Analysis of \fref{est:noisevariance}}
\label{sec:noise_estimator_analysis}

We start with the blind noise power estimator defined in \fref{est:noisevariance}. 
We have the following key result. The proof is given in \fref{app:mainresult}. 

\vspace*{-0.28cm}

\begin{theorem} \label{thm:mainresult}
Let $\bmy$ be a noisy BCG random vector with PDF as in \fref{def:noisyBCG} and with activity rate satisfying
\begin{align} 
p \leq\conditionaltextstyle \pmax \quad  \text{with} \quad  \pmax \define \frac{e^2 - 2}{2e^2 - 2}  \approx 0.421. \label{eq:probabilitycondition}
\end{align}
Let a lower bound $\LB$ and an upper bound $\UB$ be defined as follows:
\begin{align} 
\LB & \define  \conditionaltextstyle \frac{\median_Z}{ \min\left\{ \log\left(\frac{2-2p}{1-2p}\right), \log(2)(1+\SNR) \right\} } \label{eq:LB} \\
\UB & \define  \conditionaltextstyle \frac{\median_Z}{\log(2)} \! \left( \!(1-p)+\frac{p^2}{p+\SNR}  \right)\!. \label{eq:UB}
\end{align}
Then, the average noise power $\No$ satisfies\footnote{Here we simplify the notation: $\estimatedNolargeD$ converges \emph{in probability} to $\median_{Z}/\log(2)$, and strictly speaking this latter expression is the upper bound.}
\begin{align}  \label{eq:yummysandwichbound}
\LB  \leq \No \leq \UB \leq \estimatedNolargeD.
\end{align}
\end{theorem}

\fref{thm:mainresult} has the following key implications:
(i) In the large-dimension limit, the proposed blind estimate $\estimatedNo$ bounds the average noise power~$\No$ from above, i.e., we have developed a pessimistic estimator. 
(ii) If $\SNR\to0$ or $p\to0$, then $\protect{\LB = \UB = \median_Z/\log(2)}$ in \fref{eq:yummysandwichbound}, and therefore $\No=\median_Z/\log(2)$.
Thus, either for $p\to0$ or $\SNR\to0$, the proposed estimate is exact, i.e., $\plimNo \median_Z/\log(2) = \No$. 
We summarize this important insight in the following remark.

\begin{remark} \label{rem:Nolargedimension} 
In the large-dimension limit ($D\to\infty$), the proposed blind nonparametric estimate $\estimatedNo$ is pessimistic (i.e., overestimates the average noise power~$\No$), and becomes exact at low SNR or low activity rate $p$ (i.e., for sparse vectors).
\end{remark}

Next, we present bounds on the relative error of \fref{est:noisevariance}. These bounds depend on the activity rate $p$ and the SNR. The proof is given in \fref{app:errorproof}.
\begin{corollary} \label{cor:errorbound}
For $p\leq \pmax$ as in \fref{eq:probabilitycondition}, the relative error $\protect{\varepsilon \define |\estimatedNo-\No|/\No}$ of \fref{est:noisevariance} in the large-dimension limit 
is bounded as follows:
\begin{align} \label{eq:errorbound}
\!\conditionaltextstyle \frac{1}{\!1/\SNR\!+\!1/p\!+\!1\!} \leq
\displaystyle \lim_{D\to\infty}{\!\!\varepsilon} \leq 
\conditionaltextstyle \min\!\left\{\log\!\left(\!\conditionaltextstyle \frac{1\!-\!p}{1\!-\!2p}\!\right)\!, \SNR\right\}\!\!.
\end{align}
\end{corollary}

An upper bound for the relative error $\varepsilon$ can be obtained if (i)~an upper bound on the SNR is known, or (ii) an upper bound on $p$ is known, since $\log\left(\frac{1-p}{1-2p}\right)$ is nondecreasing for $p\in(0,0.5)$.
In addition, we confirm the second implication discussed below \fref{thm:mainresult}: \fref{cor:errorbound} implies that if $p\to0$ (irrespective of the SNR) or $\SNR\to0$ (irrespective of the sparsity), then the proposed estimator becomes exact, i.e., $\varepsilon=0$ and therefore $\plimNo \No$.

\subsection{Analysis of \fref{est:signalpower}}
\label{sec:signal_estimator_analysis}
For the blind estimate $\estimatedEs$ of the average signal power~$\Es$, we use the following lemma, which is derived from the fact that the entries of the vector $\bmz \define \abssquared{\bmy}$ are i.i.d. with expected value of $\Ex{}{z_d}=\Ex{}{\|\bmy\|_2^2}/D=\Es+\No$, $d=1,\ldots,D$. 

\begin{lemma} 
Let $\bmy$ be a noisy BCG random vector with PDF as in  \fref{def:noisyBCG}. Then, according to the strong law of large numbers we have
\begin{align} \label{eq:limit_of_Es}
\conditionaltextstyle \frac{1}{D} \|\bmy\|_2^2 - \No \xrightarrow[\protect{\raisebox{3pt}[0pt][0pt]{\ensuremath{ \scriptstyle{D\to\infty}}}}]{\protect{\raisebox{-0.5pt}[0pt][0pt]{\ensuremath{ \scriptstyle{a.s.}}}}} \Es. 
\end{align}
\end{lemma}
To obtain \fref{est:signalpower} in~\fref{eq:signalpowerestimator}, we construct a blind estimator of $\Es$ by taking the left side of~\fref{eq:limit_of_Es} and replacing the average noise power $\No$ with the blind estimate~$\estimatedNo$ from~\fref{est:noisevariance}. To avoid negative values of $\Es$ that have no physical meaning, we assign a value of zero to our estimate if $\|\bmy\|_2^2/D - \estimatedNo$ is negative.
Since the estimate $\estimatedNo$ overestimates the true average noise power $\No$,  the blind estimate $\estimatedEs$ in \fref{eq:signalpowerestimator}  tends to underestimate the signal power. From  \fref{thm:mainresult} it follows that for $p\to0$ or $\SNR\to0$, the blind signal power estimate~$\estimatedEs$ is exact. 

\subsection{Analysis of \fref{est:SNR}}
\label{sec:SNR_estimator_analysis}
The blind SNR estimator is obtained by simply taking the ratio of $\estimatedEs$ in \fref{eq:signalpowerestimator} and $\estimatedNo$ in~\fref{eq:noiseestimator}.
For \mbox{$D\to\infty$}, the blind signal power estimate underestimates the average signal power and the noise power estimate overestimates the average noise power, which means that the blind SNR estimate in \fref{eq:SNRestimateomg} underestimates the SNR. 
From  \fref{thm:mainresult} it follows that for $D\to\infty$ with either $p\to0$ or $\SNR\to0$ the blind SNR estimate is exact.

\subsection{Analysis of \fref{est:MSE}}
\label{sec:MSE_estimator_analysis}
In order to analyze \fref{est:MSE}, we first assume that the average noise power $\No$ is  known. For this  scenario, we can borrow the following two theorems from \cite{mirfarshbafan2019beamspace}.

\begin{theorem}[Thm.~1 of \cite{mirfarshbafan2019beamspace}]
\label{thm:MSEappox}
Consider \fref{sys:systemmodel2}.
Then, Stein's unbiased risk estimate given by
\begin{align}
  \textit{SURE} \define \, &\conditionaltextstyle \frac{1}{D}\|\fest{\bmy}-\bmy\|_2^2 - \No \notag\\
 & \conditionaltextstyle +\frac{\No}{D} \sum_{d=1}^D\left(\frac{\partial\Re\{\fest{y_d}\}}{\partial\Re\{y_d\}}+\frac{\partial\Im\{\fest{y_d}\}}{\partial\Im\{y_d\}}\right)  \label{eq:complexSURE}
\end{align}
is an unbiased estimate of the MSE so that $\Ex{}{\textit{SURE}} = \MSE.$
\end{theorem}

\begin{theorem}[Thm.~3 of \cite{mirfarshbafan2019beamspace}] \label{thm:SUREconvergence}
If $\festnoargs$ is pseudo-Lipschitz, then $\textit{SURE}$ in \fref{eq:complexSURE} converges to the MSE in the large-dimension limit, i.e., we have $\lim_{D\to\infty}\textit{SURE} = \MSE.$ 
\end{theorem}

\fref{thm:SUREconvergence} implies that if $\No$ were known perfectly, then one could perfectly estimate the MSE in the large-dimension limit without knowledge of the sparse signal vector~$\bms$. For smaller values of the dimension $D$, 
\fref{thm:MSEappox} only ensures equality in expectation (while the estimator remains MSE-optimal). Equality in expectation means that some realizations will underestimate and some realizations will overestimate the true MSE.\footnote{
We have to keep in mind that we use the estimated MSE to determine parameters in the estimation function $\festnoargs$ that minimize the MSE for each given realization of $\bmy$. 
Therefore, offsets that depend on the realization of the noisy observation $\bmy$ can be treated as a constant and thus be ignored, even if these offsets cause the MSE to take on negative values. In other words, we are not interested in the true value of the MSE, but rather in the \emph{shape} of the MSE function with respect to the parameters in $\festnoargs$.
}

\fref{est:MSE} is a blind version of SURE, in which we have replaced the true average noise power $\No$ by its estimate $\estimatedNo$. 
Consequently, for $D\to\infty$ and either $p\to0$ or $\SNR\to0$, we have that: (i) \fref{rem:Nolargedimension} states $\estimatedNo$ will be exact, from which it follows that $\estimatedMSE = \SURE$, (ii) \fref{thm:SUREconvergence} ensures $\SURE = \MSE$, and therefore (iii) \fref{est:MSE} will be exact ($\estimatedMSE=\MSE$) in this scenario.
For higher values of $p$ or $\SNR$, we know that $\estimatedNo$ tends to overestimate $\No$, but since this estimated quantity appears twice in \fref{eq:MSEnonparametricexplicitform} with different signs, we cannot derive a simple rule that states whether \fref{est:MSE} tends to underestimate or overestimate the MSE. 

\subsection{Analysis of \fref{est:sandwich}} \label{sec:sandwich_analysis}
\fref{est:sandwich} is derived as the mean of the  lower and upper bounds in \fref{eq:yummysandwichbound}, utilizing the SNR estimate from \fref{est:SNR} and an activity rate estimate $\estimatedp$ of the user's choice. \fref{est:sandwich} often improves the performance (achieves lower bias) compared to \fref{est:noisevariance}, especially at high SNR. In contrast to \fref{est:noisevariance},  we no longer know if the noise power from \fref{est:sandwich} is being overestimated or underestimated. As this estimator takes $\estimatedp$ as a parameter, it is especially useful in applications where $p$ is known a priori or bounded (e.g., in OFDM systems the number of nonzero delay taps of the channel's impulse response should not exceed the cyclic prefix length).

\subsection{Analysis of \fref{est:p}} \label{sec:p_estimator_analysis}
To estimate the activity rate, we can use the equivalence of vector norms \cite{golub13matrix} that states $\|\bmx\|_q\leq \zeronorm^{1/q-1/r} \|\bmx\|_r$ holds for any vector $\bmx\in\complexset^\zeronorm$ if $1\leq q < r$. 
In particular, it holds for a vector $\snonzeros\in\complexset^\zeronorm$ of length $\zeronorm \define \|\bms\|_0$ that contains only the nonzero entries of the sparse vector $\bms$. 
For such vector, we have that
$\|\snonzeros\|_q\leq \zeronorm^{1/q-1/r} \|\snonzeros\|_r$. 
Since the entries of $\bms$ that are zero do not contribute to these norms, we note that $\|\snonzeros\|_q = \|\bms\|_q$ and $\|\snonzeros\|_r = \|\bms\|_r$, and therefore
\begin{align} \label{eq:sparsenorminequality}
\|\bms\|_q\leq \|\bms\|_0^{1/q-1/r} \|\bms\|_r, \quad 1\leq q < r. 
\end{align}
Using \fref{eq:sparsenorminequality}, we can obtain a lower bound for the activity rate\footnote{The activity rate is $p \define \Ex{}{\|\bms\|_0}/D =\lim_{D\to\infty}{\|\bms\|_0}/D$. When $D$ is finite, we have ${\|\bms\|_0}/D \approx p$.}:  
\begin{align} \label{eq:norminequalitybound}
\conditionaltextstyle \frac{1}{D}\left(\frac{\|\bms\|_q}{\|\bms\|_r}\right)^{\frac{1}{1/q-1/r}} \leq \frac{\|\bms\|_0}{D} \approx p. 
\end{align}
 The inequality in \fref{eq:norminequalitybound} holds with equality if the nonzero entries of the signal are constant-modulus, i.e., if $|\snonzerose{d}| = |\snonzerose{d^\prime}|$, 
{$\forall\, d,d^\prime\in\{1,\ldots,L\}$.}
We obtain the blind estimator $\estimatedp(q,r)$ from the left side of \fref{eq:norminequalitybound}, by replacing $\bms$ with its noisy version $\bmy$.
With this substitution the inequality is not preserved (except if $\No=0$), but we use that definition of $\estimatedp(q,r)$ as a rough activity rate estimate instead of picking an arbitrary value.

\subsection{Analysis of \fref{est:EM}} \label{sec:EM_analysis}

\fref{est:EM} is a specialized variant of a classical EM algorithm for a two-component Gaussian mixture~\cite{hastie01a}, adapted to complex-valued and zero-mean variables.
We consider signal and noise power estimation from a noisy BCG signal as in \fref{def:noisyBCG}. 
To understand it as a Gaussian source-separation problem, we consider that each entry of $\bmy$ is a realization of either (i) just noise with distribution $\CN(0,\No)$, or (ii) signal plus noise with distribution $\CN(0,\No+\Es/p)$. 
Just-noise realizations occur with probability $1-p$, while signal-plus-noise realizations occur with probability $p$.
Using EM, we estimate the variances of the circularly-symmetric complex Gaussians $\No$ and $(\No+\Es/p)$, and mixture weights $1-p$ and $p$. 
We use our previous knowledge to set the mean of the two distributions to zero, unlike classical EM algorithms that also estimate the means. 
We make the following observations:
(i) This model allows any signal sparsity, as opposed to \fref{est:noisevariance} which assumes a maximum activity rate $\pmax$.
(ii) In the low SNR regime, EM may not be able to separate the noise and signal components, as $\No+\Es/p\approx\No$. 
(iii) The accuracy and the complexity of the algorithm will depend on the maximum number of iterations $\nitmax$, the tolerance $\tol$, the variance and weight initializations, and the noisy realization $\bmy$.

To avoid EM converging to pathological solutions with arbitrary initialization, we initialize the algorithm with the following two minimum assumptions: (i) The signal is sparse, or equivalently $\pinit\in(0,0.5)$, and (ii) the power of the entries of $\bmy$ that contain only noise is smaller than the power of the entries of $\bmy$ that contain signal plus noise, or equivalently $\Noinit\leq \|\bmy\|_2^2/D$.
This translates to initializing the Gaussian mixture variances $v$ and weights $w$ with $v_a \leq \|\bmy\|_2^2/D$, $w_b\in(0,0.5)$, $w_a = 1-w_b$, and $v_b = v_a+(\|\bmy\|_2^2/D-v_a)/w_b$. We verify that for this initialization, the average power of the mixture is $w_a v_a + w_b v_b = \|\bmy\|_2^2/D$, as expected.

\newcommand{\trials}{10000}
\newcommand{\figW}{0.45\textwidth}
\newcommand{\showvariance}{1}
\newcommand{\updatedfigures}{figures}
 
\section{Synthetic Results} \label{sec:synthetic_results}

We now characterize the accuracy of the estimators proposed in \fref{sec:nonparametricestimators}.
We use the sparse signal model in \fref{def:noisyBCG}. Without loss of generality, we fix the noise power to $\No=1$, while varying the signal power $\Es$, the activity rate $p$, and the dimension $D$ of the vectors.  
For different sets of parameters, we perform Monte--Carlo simulations with \SI{\trials}{} trials. 
In the plots, the thicker line with markers shows the average performance of an estimator, while the shaded area shows the region closer than one standard deviation away from the mean performance, a measure of the 
{\emph{precision} of the estimator}.

\subsection{Evaluation of the Noise Power, Signal Power, SNR, and Activity Rate Estimators} 
\fref{fig:estNo} shows the effect of the SNR on the performance of the proposed blind nonparametric estimator $\estimatedNo$ from \fref{eq:noiseestimator} and the proposed blind parametric estimate $\estimatedNosandwich$ from \fref{eq:sandwich_estimator}, for which we only include results using $\estimatedp(q,r)$ with $q=1$ and $r=\infty$ for the activity rate estimate, as these parameters showed the best performance in our simulations, outperforming other values of~$p$ and~$q$, and a fixed-value of $\estimatedp=0.25$ which is the center of the simulated range $p\in(0,0.5)$.
We also simulate the baseline EM estimate $\estimatedNoEM$ described in \fref{est:EM}, initialized with $\Noinit=0.4\|\bmy\|_2^2/D$ and $\pinit=\estimatedp(1,\infty)$, a maximum of $\nitmax=30$ iterations and early stopping if the total parameter change is below $\tol=0.1$\%.
As a baseline, we plot the genie-aided estimator $\sampleNo \define \frac{1}{D}\|\bmn\|_2^2$ that has separate knowledge of~$\bmn$ and the reference parameter~$\No$.

\begin{figure*}[p]
	\centering
	\subfigure[$p=0.1$, $D=64$\label{fig:estNoa}]{\includegraphics[width=\figW]{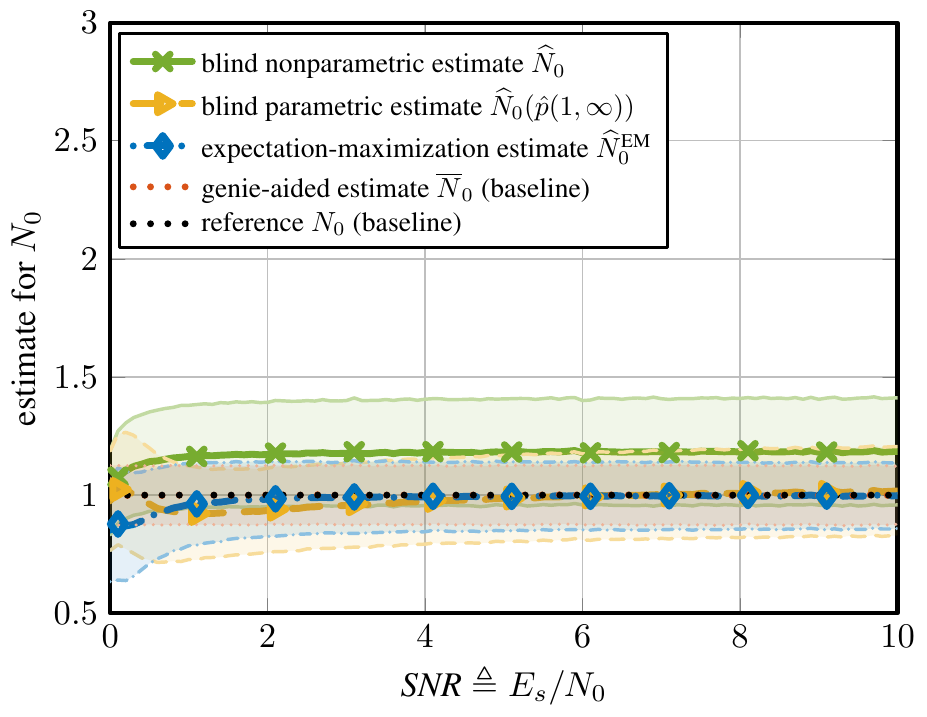}} 
	\hfill
	\subfigure[$p=0.1$, $D=256$\label{fig:estNob}]{\includegraphics[width=\figW]{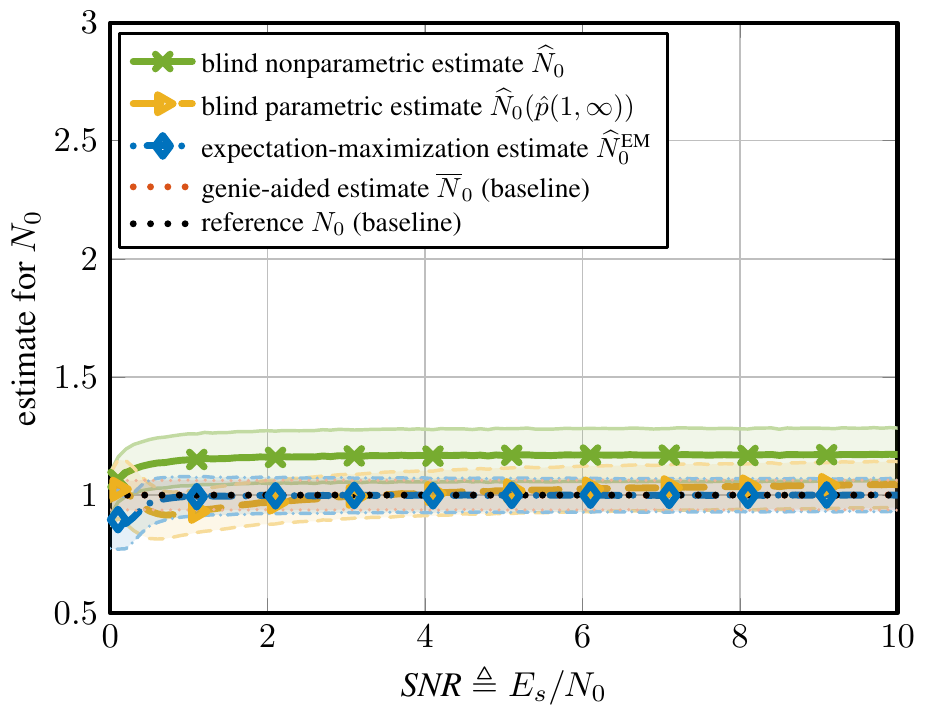}}
	\subfigure[$p=0.4$, $D=64$\label{fig:estNoc}]{\includegraphics[width=\figW]{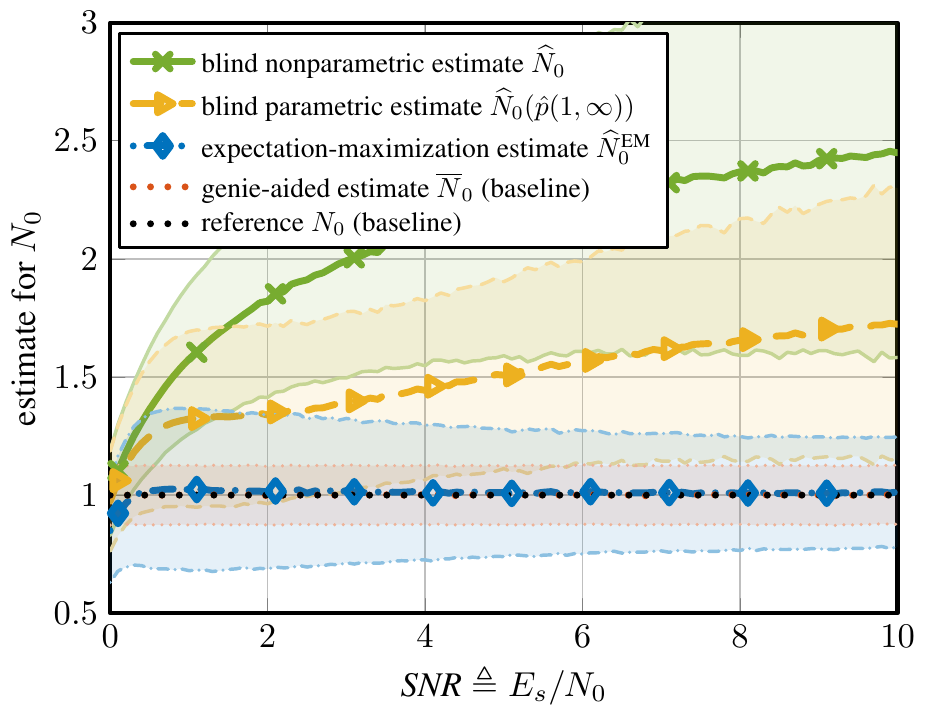}}
	\hfill
	\subfigure[$p=0.4$, $D=256$\label{fig:estNod}]{\includegraphics[width=\figW]{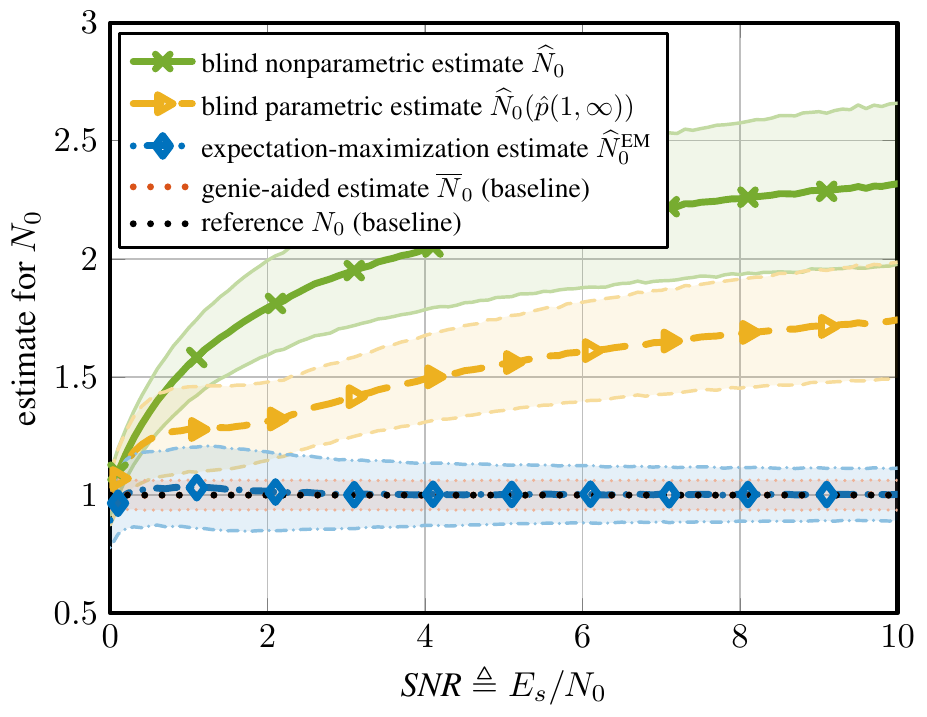}}
	\caption{Effect of varying the SNR on the proposed low-complexity blind {noise power estimators}, for different values of the activity rate $p$ and the dimension~$D$. Blind estimators are more accurate for smaller values of $p$, and the precision of all estimators is higher for larger values of $D$.}	
	\label{fig:estNo}	
\end{figure*}

\fref{fig:varyingSNR} shows the effect of the SNR on the performance of the proposed signal power and SNR estimators for an activity rate of $p=0.1$ and a dimension of $D=64$. 
In this case, $\estimatedSNR{\>\!}^\text{EM} \define \estimatedEsEM/\estimatedNoEM$, the genie-aided estimators that have separate knowledge of $\bms$ and~$\bmn$ are $\sampleEs \define \frac{1}{D}\|\bms\|_2^2$ and $\protect{\sampleSNR \define {\sampleEs}/{\sampleNo}}$, and the reference parameters are $\Es$ and $\SNR \define \Es/\No$.

\begin{figure*}[p]
	\centering
	\subfigure[Average signal power]{\includegraphics[width=\figW]{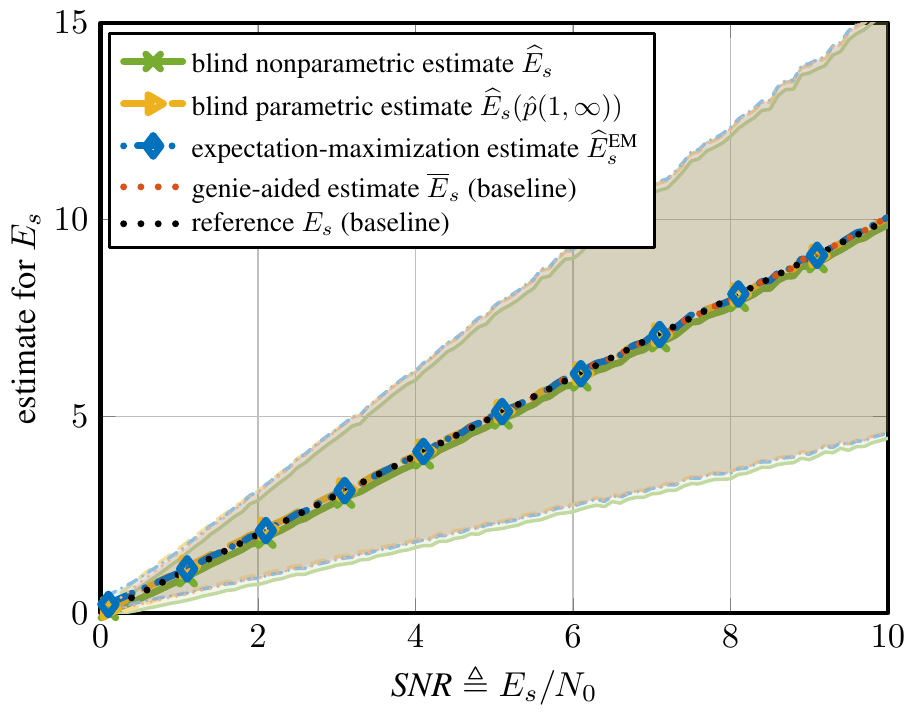}} 
	\hfill
	\subfigure[Signal-to-noise ratio]{\includegraphics[width=\figW]{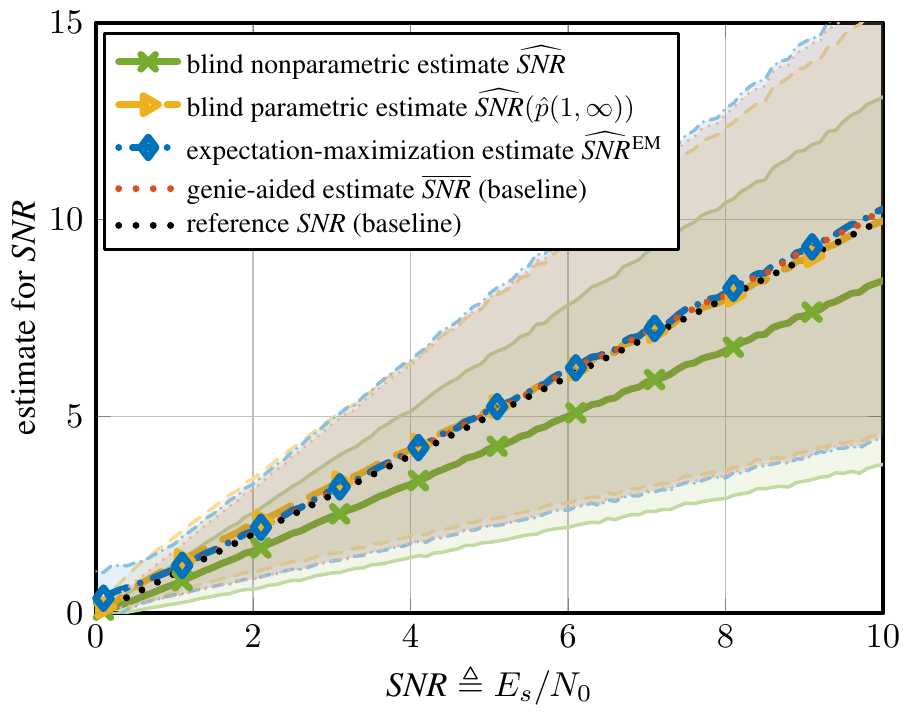}}
	\caption{Effect of varying the SNR on the proposed low-complexity blind {signal power and SNR} estimators, for a signal of dimension $D=64$ with $10\%$ nonzero entries.}
	\label{fig:varyingSNR}
\end{figure*}

\begin{figure*}
	\centering
	\subfigure[$\SNR=0.5$]{\includegraphics[width=\figW]{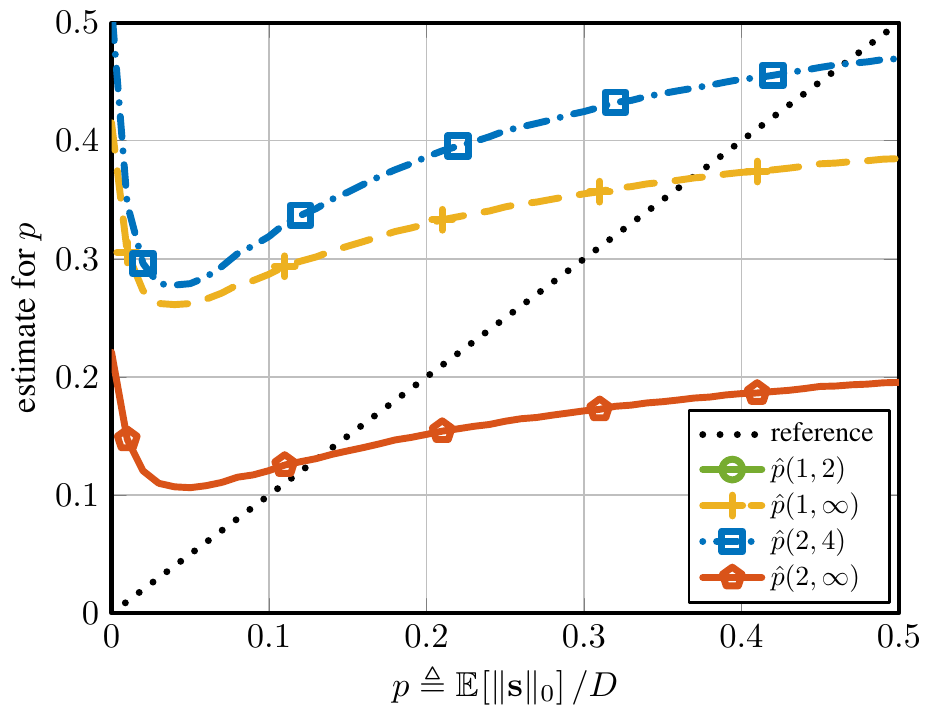}} 
	\hfill
	\subfigure[$\SNR=10$]{\includegraphics[width=\figW]{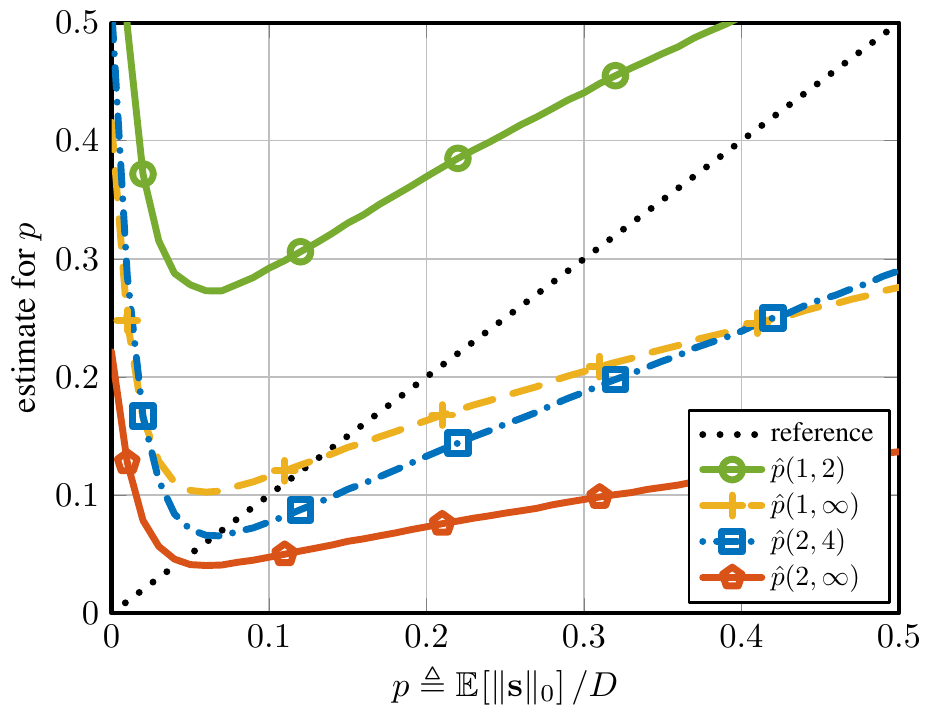}}
	\caption{Average accuracy of the blind parametric activity rate estimator $\estimatedp(q,r)$ for different combinations of $q$ and $r$ with vectors of length $D=64$. {For $\SNR=0.5$, the results for $\estimatedp(1,2)$ are not visible as the values exceed $0.5$.} While the overall accuracy of this family of estimators appears to be low, the estimate  obtained by $\estimatedp(1,\infty)$ combined with \fref{est:sandwich}, for example, shows superior performance compared to using the mean activity rate in the considered interval  (i.e., $0.25$).}
	\label{fig:pest}
\end{figure*}

From Figures~\ref{fig:estNo} and~\ref{fig:varyingSNR}, we observe the following facts about the blind nonparametric estimators:
(i) For sparse vectors ($\protect{p=0.1}$), our estimators have a precision comparable to that of the genie-aided estimators even for a small sample size of $D=64$. 
(ii) The {precision} of all considered estimators decreases as~$D$ increases. 
(iii) As predicted by our theory, the average noise power is overestimated while the signal power and SNR are underestimated. 
(iv) At low SNR, the {median-based} estimators for these three quantities become exact.
We also observe that the proposed blind parametric estimate $\estimatedNo(\estimatedp)$ with $\estimatedp = \estimatedp(1,\infty)$ is more accurate than the blind nonparametric estimate $\estimatedNo$ at high SNR. 
However, $\estimatedNosandwich$ has fewer theoretical guarantees and is not an upper bound on~$\No$.

\fref{fig:pest} shows the accuracy of  the blind, parametric activity rate \fref{est:p}. We see that  at low and high SNR, $\estimatedp(1,2)$ tends to overestimate $p$ while $\estimatedp(2,\infty)$ tends to underestimate~it. Overall, $\estimatedp(1,\infty)$ results in the best performance  when combined with \fref{est:sandwich}. Admittedly, this is only a rough estimator and we include it as an example of  what could be plugged into \fref{est:sandwich} or \fref{est:EM}. Nonetheless, we emphasize that side information about the signal's sparsity should be utilized whenever available.

In comparison with EM (cf.~Figures~\ref{fig:estNo} and~\ref{fig:varyingSNR}), our methods provide a less-accurate estimate at higher SNR, but require significantly lower complexity. 
The complexity of the baseline EM algorithm (in terms of the number of operations such as real-valued additions, real-valued multiplications, and exponentials) is more than  $\nit(16D+12)+3D$ operations---with early stopping, the average number of iterations observed in our simulations ranges from $\nit=8$ to $\nit=28$ depending on the SNR. 
In contrast, our proposed median-based noise estimator has an average complexity of no more than $7.7D+9$ operations if the median is computed using quickselect~\cite{tibshirani09a},
and avoids the evaluation of operations such as  exponentials and divisions. 
Hence, our proposed blind estimator is more than $17\times$ less complex than the baseline EM algorithm, which renders our method suitable (i) for low-complexity parameter estimation and (ii) as a potential initializer for EM-based estimators. 

\subsection{Accelerated Convergence of EM Using Median-Based Initialization} \label{sec:accelerated_convergence}
\renewcommand{\showvariance}{1}
\begin{figure*}
	\centering
	\subfigure[$p=0.1$, $\SNR=0.1$\label{fig:p0p1SNR0p1}]{\includegraphics[width=\figW]{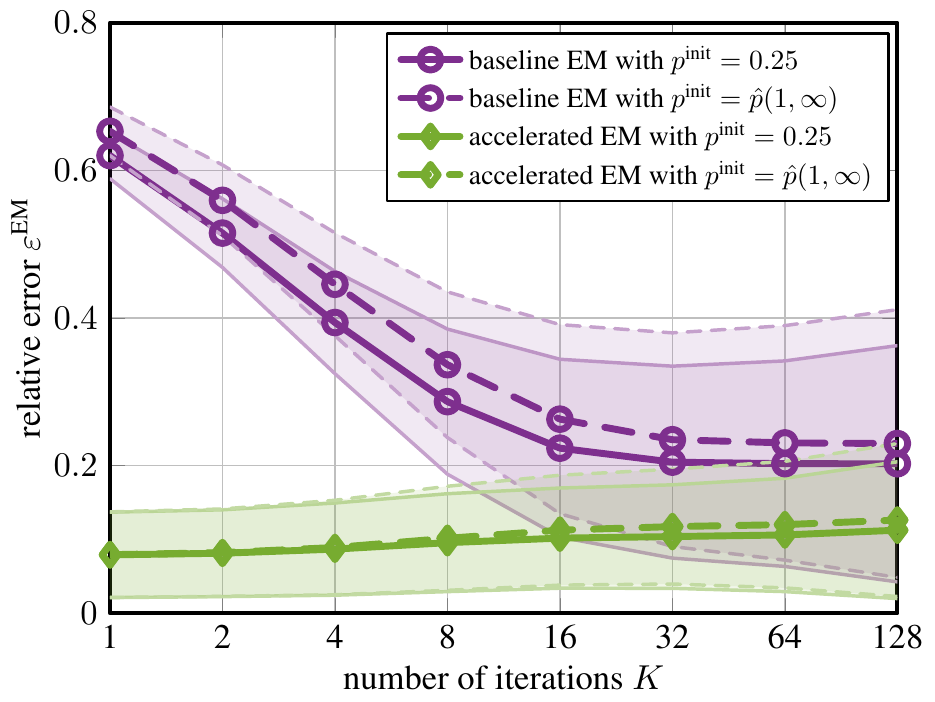}} 
	\hfill
	\subfigure[$p=0.1$, $\SNR=5$\label{fig:p0p1SNR5}]{\includegraphics[width=\figW]{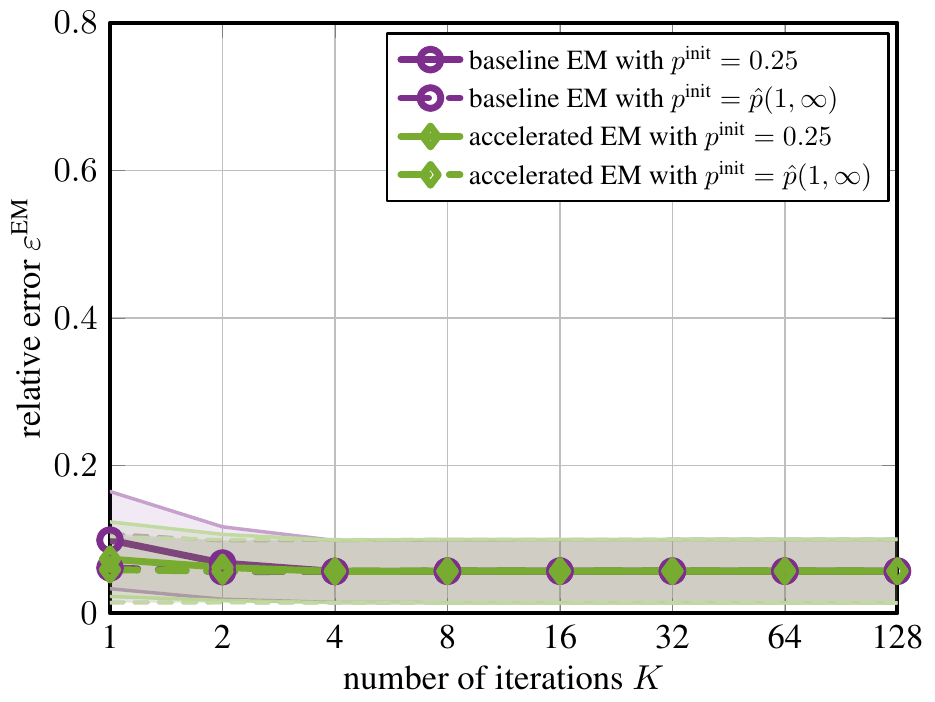}}
	\subfigure[$p=0.4$, $\SNR=0.1$\label{fig:p0p4SNR0p1}]{\includegraphics[width=\figW]{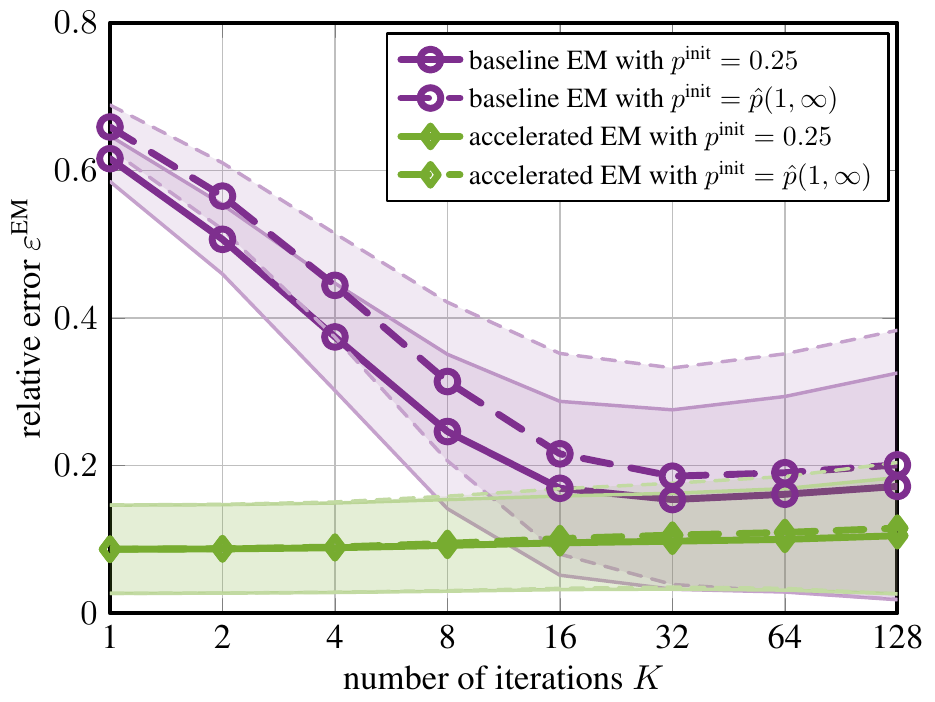}}
	\hfill
	\subfigure[$p=0.4$, $\SNR=5$\label{fig:p0p4SNR5}]{\includegraphics[width=\figW]{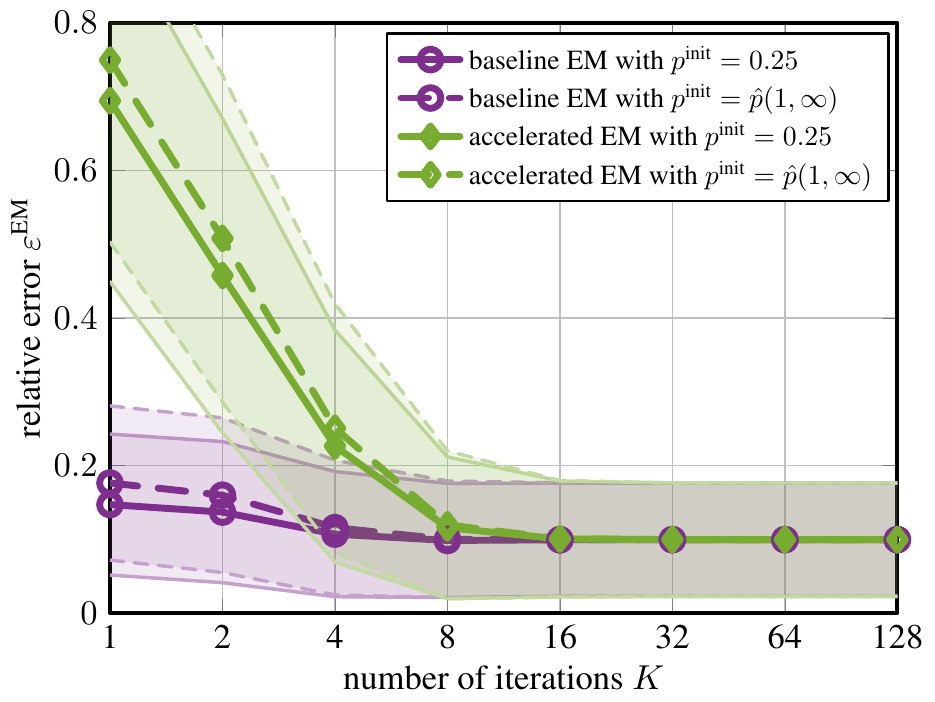}}
	\caption{Relative error $\varepsilon^\text{EM} \define |\estimatedNoEM-\No|/\No$ of the baseline and accelerated EM noise power estimators vs. the number of iterations $\nit$, for different values of SNR and activity rate $p$. Initializing with the low-complexity estimator $\estimatedNo$ accelerates convergence.}
	\label{fig:iterations}	
\end{figure*}

\fref{fig:iterations} shows the effect of initialization on the EM algorithm. 
To study the rate of convergence, we disable early stopping by setting $\tol=0$ so that the number of iterations is always $\nit=\nitmax$, and plot the relative error $\varepsilon^\text{EM} \define |\estimatedNoEM-\No|/\No$ as we vary $\nitmax$.
We compare the convergence of (i) the accelerated EM algorithm (diamond markers) which is initialized with the blind nonparametric estimate $\estimatedNo$, and (ii) the baseline EM algorithm (circular markers) initialized with a fixed initialization of $\SNR=5$, which corresponds to setting the noise power to $1/6$ of the received power $\|\bmy\|_2^2/D$.
We simulated various values of~$p$ and $\SNR$, and picked four examples that are representative.
At low SNR or high sparsity (low $p$), the accelerated EM algorithm converges already in the first iteration.
In contrast, the baseline EM algorithm converges in more than 16 iterations in some cases.
The only case we observe the baseline to outperform the accelerated EM algorithm is in \fref{fig:p0p4SNR5}, in which 
(i) the baseline has advantage since $\frac{1}{6}\|\bmy\|_2^2/D$ coincides exactly with the true value of $\No$, and 
(ii) the SNR is high and the sparsity is low, 
making it the worst case for the $\estimatedNo$ estimate used by the accelerated EM.
We also examine the effect of initializing the activity rate with (i) a fixed value of~0.25, versus (ii) the blind parametric estimate $\estimatedp(1,\infty)$, and we observe no significant difference, especially for the preferred accelerated EM algorithm; however, as $\estimatedp(1,\infty)$ showed superior performance than a fixed value when used in the parametric noise power estimator $\estimatedNo(\estimatedp)$, we prefer $\estimatedp(q,r)$ when no side information about the signal's sparsity is~available.

\subsection{Evaluation of the MSE Estimator}
\begin{figure*}
	\centering
	\subfigure[$\SNR=0.5$\label{fig:MSE_SNR0p5}]{\includegraphics[width=\figW]{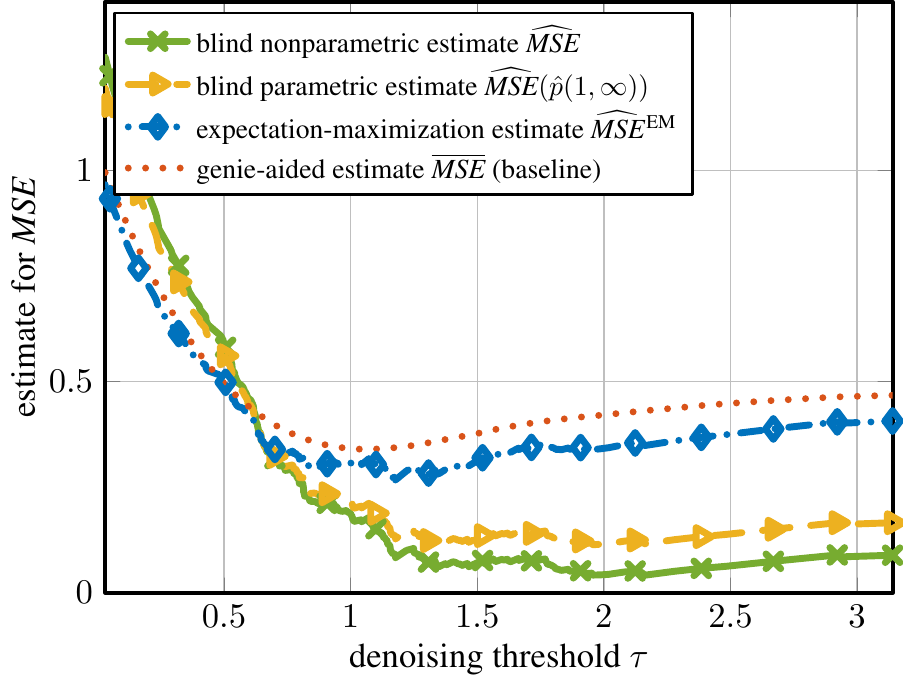}}
	\hfill
	\subfigure[$\SNR=10$\label{fig:MSE_SNR10}]{\includegraphics[width=\figW]{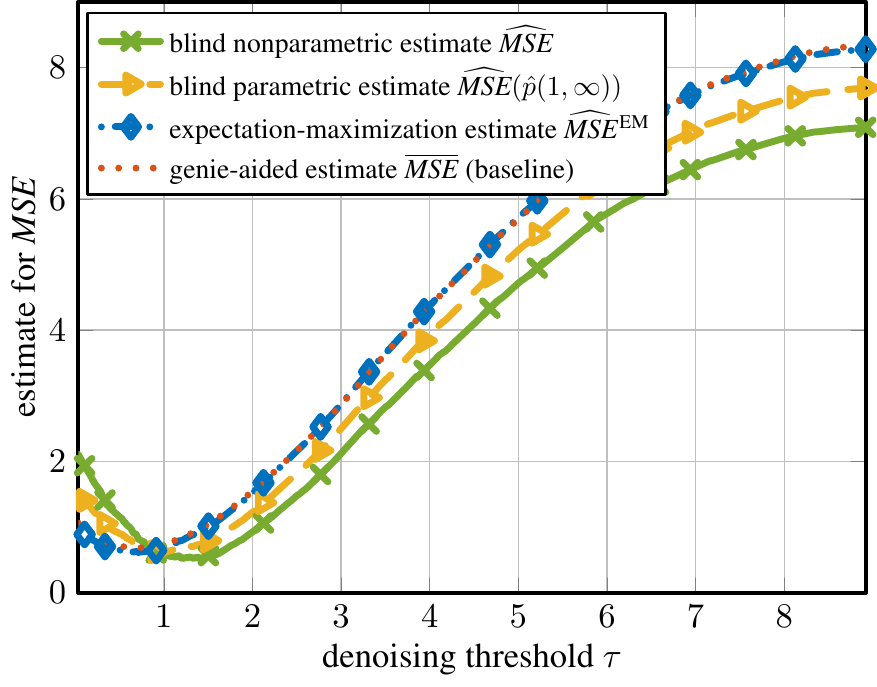}}
	\caption{Examples of the estimated MSE versus $\tau$ given one realization of $\bmy$ at (a) {low SNR} and (b) {high SNR}.}
	\label{fig:MSE}
\end{figure*}

To evaluate the performance of the MSE estimator, we consider \fref{sys:systemmodel2} with $\festnoargs$ being the soft-thresholding function defined as
\begin{align} \label{eq:soft-thresholding}
\fest{x;\tau} \define \left\{
\begin{array}{ll}
           \frac{x}{|x|}\max\{|x|-\tau,0\} & x \neq 0\\
           0 & x = 0,
\end{array}\right.
\end{align}
where {the} denoising threshold is a real number $\tau\geq0$. 

\fref{fig:MSE} shows two realizations of the estimated MSE as a function of the tuning parameter $\tau$. The only reference in this case is the genie-aided estimator $\sampleEo \define \frac{1}{D} \|\eta(\bmy;\tau)-\bms\|_2^2$. 
We picked two examples that are representative of what we have observed through multiple experiments with different system parameters to illustrate the following observations: 
(i) If the MSE function has a pronounced minimum as in \fref{fig:MSE_SNR10}, then the value of $\tau$ that minimizes the blind estimate tends to be very close to the value that minimizes the genie-aided MSE function. 
(ii) If the MSE function has a less pronounced minimum as in \fref{fig:MSE_SNR0p5}, then the value of $\tau$ that minimizes the blind estimate may be far from the value that minimizes the genie-aided MSE function. 
%
{In spite of that}, because the MSE function is flat near the minimum, the genie-aided MSE function evaluated at these two values of $\tau$ returns values that are similar.
In other words, (i) and (ii) summarize our observations that our algorithm finds a near-optimal (sub-optimal) denoising threshold $\tau$ when the MSE of the denoised channel is (not) sensitive to $\tau$.
Note that here we have only picked two representative realizations; in \fref{sec:channel_denoising}, we validate our estimator with quantitative results by showing the denoising performance averaged over many realizations.

\section{Applications to Nonparametric Channel-Vector Denoising} \label{sec:channel_denoising}
We show three applications in wireless systems, in which the quality of channel estimates is essential for data detection. 
Concretely, we show that our algorithms can be applied to adaptively denoise pilot-based channel estimates, resulting in a reduced (improved) bit-error-rate (BER).

\renewcommand{\trials}{10000} 

\subsection{Infinite-Resolution Massive Multiuser MIMO System} \label{sec:infinite-res}
We start with an application of \fref{est:MSE} for beamspace channel estimation. 
As in \cite{ghods19a}, we simulate an uplink massive multiuser (MU) MIMO system in which $\U=8$ single-antenna user equipments (UEs) transmit channel-estimation pilots and data to a {basestation (BS)} equipped with a uniform linear array of $\B=128$ antenna elements. 
The UEs are randomly placed with a uniform distribution in a \ang{120} circular sector around the BS, with a minimum distance of $10$\,m and maximum distance of $110$\,m from the BS. A minimum angular separation of \ang{4} between UEs is enforced. 
We assume UE-side perfect power control (UEs adjust their transmit power so that the received power at the BS is equal for all UEs), and we ignore quantization at transmitter and receiver sides, assuming infinite-resolution signals. 

We simulate a noiseless channel matrix $\bH\in\complexset^{\B\times\U}$ using line-of-sight (LoS) realizations from the mmMAGIC QuaDRiGa model~\cite{jaeckel2019quadriga} with a carrier frequency of $f_c=60$\,GHz. 
Each complex-valued entry $\bH_{d,u}$ of the channel matrix contains the attenuation and phase between the $u$th UE and the $d$th BS antenna.
For the channel estimation step, the UEs transmit orthogonal pilots. The maximum likelihood (ML) estimate of the channel matrix is obtained by right-multiplying the (noisy) received pilot sequence with the inverse of the orthogonal pilot matrix, resulting in
\begin{align}
\HML = \bH + \NCE, \label{eq:HML}
\end{align}
where $\bH\in\complexset^{\B\times\U}$ is the antenna-domain channel matrix, $\NCE\in\complexset^{\B\times\U}$ is complex Gaussian channel estimation noise with power $\NoCE$ per complex entry, and $\HML\in\complexset^{\B\times\U}$ is the ML channel estimate, which is a noisy observation of $\bH$.
The beamspace representation of the ML estimate is  obtained by taking a spatial Fourier transform across the antenna array resulting in
\begin{align}
\HMLbeamspace = \Hbeamspace + \NCEbeamspace. \label{eq:HMLbeamspace}
\end{align}
Here, beamspace-domain quantities are designated by a tilde.
Then, $\Hbeamspace=\bF\bH$ is the beamspace channel matrix, $\protect{\NCEbeamspace=\bF\NCE}$ has the same distribution as $\NCE$ as the discrete Fourier transform matrix $\bF$ is unitary, and $\HMLbeamspace$ is the beamspace ML channel estimate, which is a noisy observation of $\Hbeamspace$.
Column indices of $\Hbeamspace$ correspond to UEs, while row indices correspond to different angles-of-arrival to the BS. Since electromagnetic waves at high carrier frequencies experience strong attenuation, typical mmWave channels consist only of a small number of dominant propagation paths arriving at the BS.
Thus, each column of $\Hbeamspace$ (which is the beamspace \emph{channel vector} of one UE) will be approximately sparse, with many entries being close to zero.

By writing each column of \fref{eq:HMLbeamspace} as an independent equation, we can express the channel estimation problem in the form of \fref{sys:systemmodel1}, that is, each beamspace channel vector (that contains only few nonzero entries) corresponds to a sparse signal $\bms$. 
The sparsity property implies that we can perform denoising to improve the ML channel estimate.
After channel estimation, all UEs transmit data simultaneously using uncoded 16-QAM symbols and the BS performs data detection using the estimated channel vectors and linear minimum MSE equalization.

\begin{figure*}
	\centering
	\subfigure[Uncoded BER vs.~SNR]{\includegraphics[width=\figW]{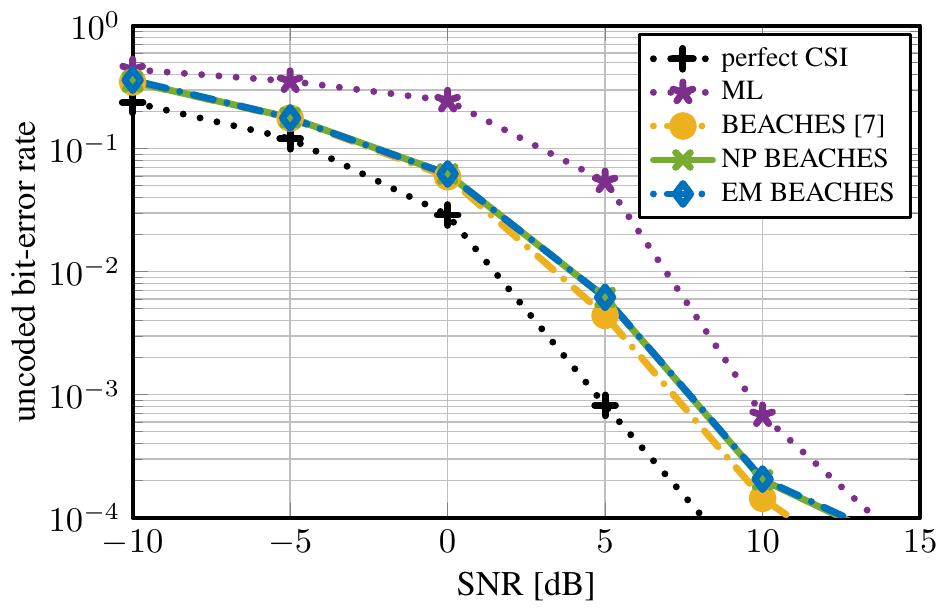}}
	\hfill
	\subfigure[MSE vs.~SNR]{\includegraphics[width=\figW]{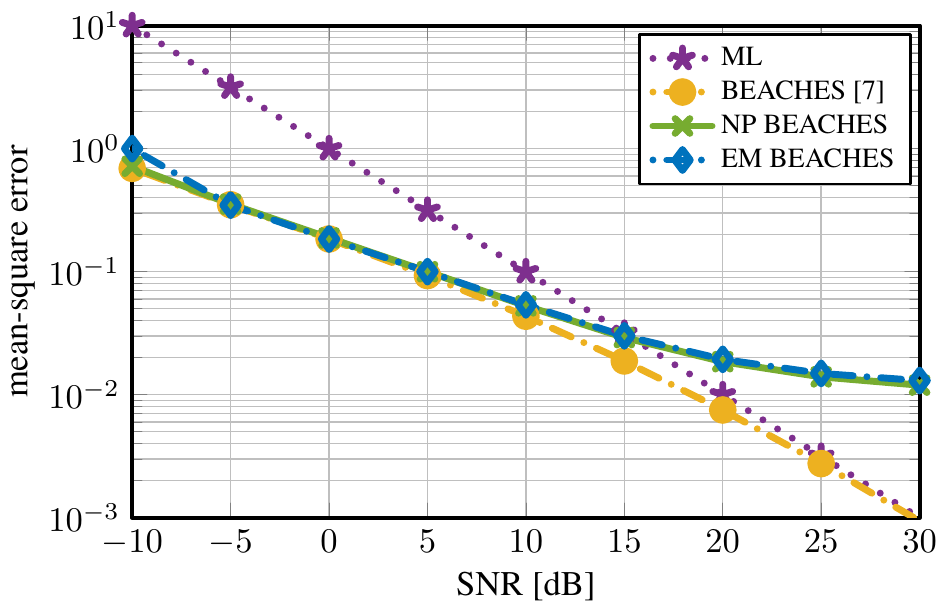}}
	\caption{Uncoded BER (a) and MSE (b) of mmWave channel estimation methods, including the nonparametric (NP) BEACHES variant which estimates the noise power and denoising parameter directly from the receive vector.}
	\label{fig:inf_res}
\end{figure*}

\fref{fig:inf_res} shows simulation results for \SI{\trials}{} Monte--Carlo trials.
For different channel estimation methods, we compute the MSE of the channel estimates and the resulting BER. 
We simulate beamspace channel estimation (BEACHES) as in~\cite{ghods19a}, which denoises the columns of $\HMLbeamspace$ in~\fref{eq:HMLbeamspace} by applying the soft-thresholding function in \fref{eq:soft-thresholding};
the thresholding parameter $\tau$ is adaptively selected for each noisy observation by minimizing SURE using an $\setO(D\log(D))$ algorithm that assumes perfect knowledge of the average noise power~$\NoCE$.
We compare this to NP BEACHES, a new nonparametric BEACHES variant which also applies soft-thresholding to the columns of $\HMLbeamspace$, but uses the (nonparametric) threshold $\tau$ that minimizes $\estimatedMSE$ as in \fref{est:MSE}; since $\estimatedMSE$ is a nonparametric version of SURE, NP BEACHES does not require knowledge of $\NoCE$.
In addition, we include a variant that we call EM BEACHES, which uses a version of $\estimatedMSE$ in which $\estimatedNo$ in \fref{eq:MSEnonparametricexplicitform} is replaced by $\estimatedNoEM$ from \fref{est:EM}; for $\estimatedNoEM$, we use $\Noinit=0.4\|\bmy\|_2^2/D$ and $\pinit=\estimatedp(1,\infty)$, a maximum of $\nitmax=30$ iterations and early stopping if the total parameter change is below $\tol=0.1$\%.
The three versions of BEACHES as described above, after denoising the beamspace channel vectors, use the inverse Fourier transform to obtain an antenna-domain channel estimate to be used for data detection.
As a reference, we show the performance of perfect channel state information (CSI) that uses the ground truth (noiseless) channel matrix $\bH$, and ML estimation that simply takes the noisy observation~$\HML$ in \fref{eq:HML} as the estimate.

From \fref{fig:inf_res}, we observe that NP BEACHES achieves virtually the same performance as the original BEACHES algorithm (which requires knowledge of~$\NoCE$), except at high SNR where \fref{est:noisevariance} tends to overestimate $\NoCE$. We reiterate that NP BEACHES requires no parameters and exhibits the same low complexity of $\setO(D\log(D))$ as the original BEACHES algorithm, because the latter already sorts the entries of $\abssquared{\bmy}$, which we can reuse to compute the median in \fref{est:noisevariance}.
We observe that EM BEACHES achieves higher (worse) MSE at low SNR and does not outperform NP BEACHES at higher SNR.

In summary, denoising methods can significantly improve the ML channel estimate. All three BEACHES variants achieve similar BER performance. However, BEACHES needs knowledge of the noise power and EM BEACHES exhibits higher complexity than our nonparametric estimate, which renders NP BEACHES the preferable denoising method in this application scenario.

\subsection{Low-Resolution Massive Multiuser MIMO System} \label{sec:1-bit}
Next, we consider the same uplink massive MU-MIMO system as \fref{sec:infinite-res}, but in this case each radio-frequency (RF) chain at the BS is equipped with a pair of 1-bit analog-to-digital converters (ADCs) to quantize the in-phase and quadrature baseband signals. Each RF chain applies a quantization function $Q(x) \define \sign\left(\Re\{x\}\right) + j \sign\left(\Im\{x\}\right)$ to the baseband signal, where $j^2=-1$. For simplicity, we assume that the pilot matrix is an identity, i.e., each UE has a dedicated time slot to transmit one pilot while all other UEs are silent. The receive pilots then correspond to the 1-bit version of the ML channel estimate, which we call 1-bit ML\footnote{$\HMLonebit$ is simply the 1-bit version of $\HML$, not to be confused with the maximum likelihood channel estimate given a one-bit observation.}
\begin{align}
\HMLonebit = Q\left(\bH + \NCE\right). \label{eq:HML1b}
\end{align}
Here, quantization happens in the antenna domain and yet, when the quantized noisy channel is converted to beamspace, the sparse structure that is present in infinite-resolution beamspace channel vectors is also present in the coarsely quantized beamspace channel vectors. Thus, 
\begin{align}
\HMLonebitbeamspace = \bF Q\left(\bH + \NCE\right) \label{eq:HML1bitbeamspace}
\end{align}
has sparse columns that can be denoised.
For more details on the validity~of this statement, see~\cite{gallyas20a} where  $\HMLonebitbeamspace$ was decomposed in a linear combination of $\Hbeamspace$ plus a residual.

\begin{figure*}
	\centering
	\subfigure[Uncoded BER vs.~SNR]{\includegraphics[width=\figW]{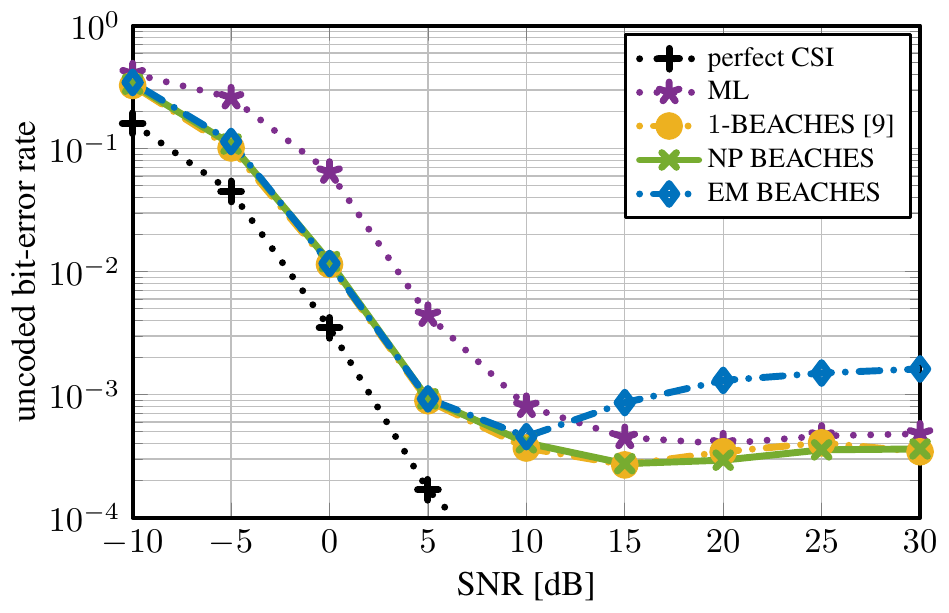}}
	\hfill
	\subfigure[MSE vs.~SNR]{\includegraphics[width=\figW]{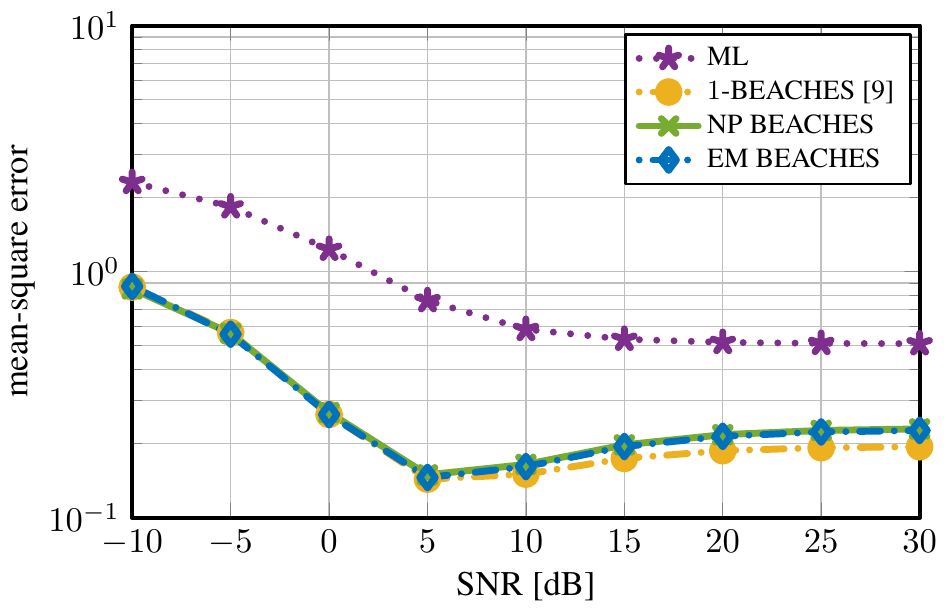}}
	\caption{Uncoded BER (a) and MSE (b) of channel estimation methods in a mmWave system with 1-bit quantization.}
	\label{fig:one_bit_res}
\end{figure*}

\fref{fig:one_bit_res} shows simulation results for \SI{\trials}{} Monte--Carlo trials.
For different channel estimation methods, we compute the MSE and BER. All UEs simultaneously transmit uncoded QPSK symbol, and the BS uses the estimated channels in order to perform 1-bit Bussgang linear minimum MSE equalization as described in \cite{nguyen19a}.

We simulate $1$-BEACHES as in~\cite{gallyas20a}. This denoising algorithm decomposes \fref{eq:HML1bitbeamspace} as $\HMLonebitbeamspace = \Hbeamspace + \widetilde{\bQ}$, where $\widetilde{\bQ}$ represents the equivalent noise-plus-quantization error 
{with average power}
$Q_0 = 2+\Es-4\Es/\sqrt{\pi(\Es+\No)}$ per entry~\cite{gallyas20a}. 
The $1$-BEACHES algorithm denoises the columns of $\HMLonebitbeamspace$ with the threshold $\tau$ that minimizes SURE, assuming perfect knowledge of $Q_0$.
We also use the nonparametric algorithms NP BEACHES and EM BEACHES (described in \fref{sec:infinite-res}) to denoise the columns of $\HMLonebitbeamspace$. 
After denoising the beamspace channel vectors, these three BEACHES variants use the inverse Fourier transform to obtain an antenna-domain channel estimate. 
We compare these estimators with $\HMLonebit$ from \fref{eq:HML1b}, and with the perfect CSI estimate that uses the ground truth $\bH$ as the channel estimate. 

Since NP BEACHES uses the median-based noise estimate (which in this case estimates the effective ``noise'' floor that includes quantization errors), it is robust to outliers and is able to achieve MSE and BER performance very close to $1$-BEACHES that has perfect knowledge of the noise-plus-quantization power. 
The EM estimator, however, strongly relies on the distribution of the noise and signal being Gaussian. 
{Here}, the signal is a realistic channel vector which is not Gaussian; 
more importantly, $\widetilde{\bQ}$ contains the effect of noise but also quantization error, which means the equivalent noise also deviates from a Gaussian distribution. We attribute the higher (worse) BER of EM BEACHES to these two factors. 
We note that $1$-BEACHES is designed specifically for 1-bit quantization and that the expression for $Q_0$ (which requires knowledge of the noise power and the signal power) would be different if the ADCs use a different number of bits. 
In contrast, our nonparametric denoiser is agnostic to the quantizer's resolution and automatically determines the power of the noise plus quantization, as long as the signal is approximately sparse and the noise is approximately Gaussian. 

\subsection{Cell-Free Communication System} \label{sec:cell-free}
\renewcommand{\U}{U}
\renewcommand{\B}{D}
We simulate an uplink cell-free communication system with $\U=16$ single-antenna UEs and $\B=256$ single-antenna BSs. The UEs and BSs are randomly placed with a uniform distribution in a square with \SI{1}{\kilo\meter^2} area. The UEs transmit orthogonal pilots followed by QPSK data. All of the UEs transmit simultaneously and the received signal at all the BSs is processed at a central processing unit (CPU) that performs channel estimation and linear minimum MSE detection.

We simulate a cell-free channel matrix $\bH$ using the model proposed by~\cite{ngo17a}, with parameters as in~\cite{song20a} but without power control and with a transmit power of $12.5$\,mW per UE.
As in \fref{eq:HML}, the ML estimate of the channel matrix is obtained by right-multiplying the pilot sequence received at the CPU with the inverse of the orthogonal pilot matrix (we used a Hadamard pilot matrix), resulting in
\begin{align}
\HML = \bH + \NCE. \label{eq:HMLcellfree}
\end{align}
The columns of $\bH$ (or \emph{channel vectors}) contain the attenuations and phases between one UE and all BSs. For each UE, the BSs that have LoS or are closer to this UE will receive significantly higher power than the other BSs that are not nearby. This means that in the cell-free system, the channel vectors are approximately sparse~\cite{gholamipourfard21a} and the ML estimate can be denoised. 
{Although the thermal noise variance at different basestations may differ, we assume i.i.d.\ noise in this paper.}

\fref{fig:cell_free} shows the results of \SI{\trials}{} Monte--Carlo trials. On the left, we plot the CDF of the MSE of the channel estimates, and on the right, the CDF of the root-mean-squared-symbol-error (RMSSE). The RMSSE is a measure of how far  the expected QPSK symbol is from the received data symbol after equalization with the channel estimates, and can be seen as equivalent to the error-vector-magnitude (EVM) for one UE. 

In \fref{fig:cell_free}, we observe a clear MSE improvement of the three denoising algorithms over the ML estimate: For a given value $x$, there are more realizations of channel estimates whose MSE is smaller than $x$ for denoised channels than for ML. 
The fact that denoising improves the channel estimates is reflected in the RMSSE, since equalization is more effective and the obtained symbols are closer to the expected constellation points. 
We consider the RMSSE requirement of \SI{17.5}{\percent} for QPSK from \cite[Table 6.5.2.2-1]{3gpp19a}. 
The probability that a UE meets the requirement grows from \SI{0.43}{} with ML channel estimation, to \SI{0.59}{} with NP BEACHES or EM BEACHES denoising, an increase of \SI{0.16}{}. BEACHES with perfect knowledge of the noise power has a slight additional advantage, with a probability of meeting the requirement of \SI{0.66}{}. 

\begin{figure*}
\centering
\subfigure[CDF of the mean-square error (MSE).]{\includegraphics[width=\figW]{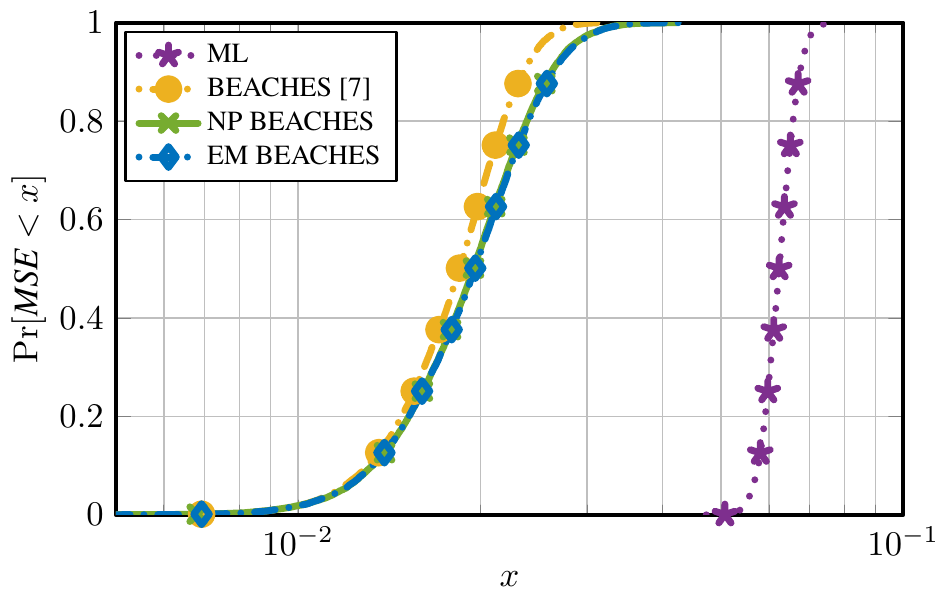}}
\hfill
\subfigure[CDF of the root-mean-squared-symbol error (RMSSE).]{\includegraphics[width=\figW]{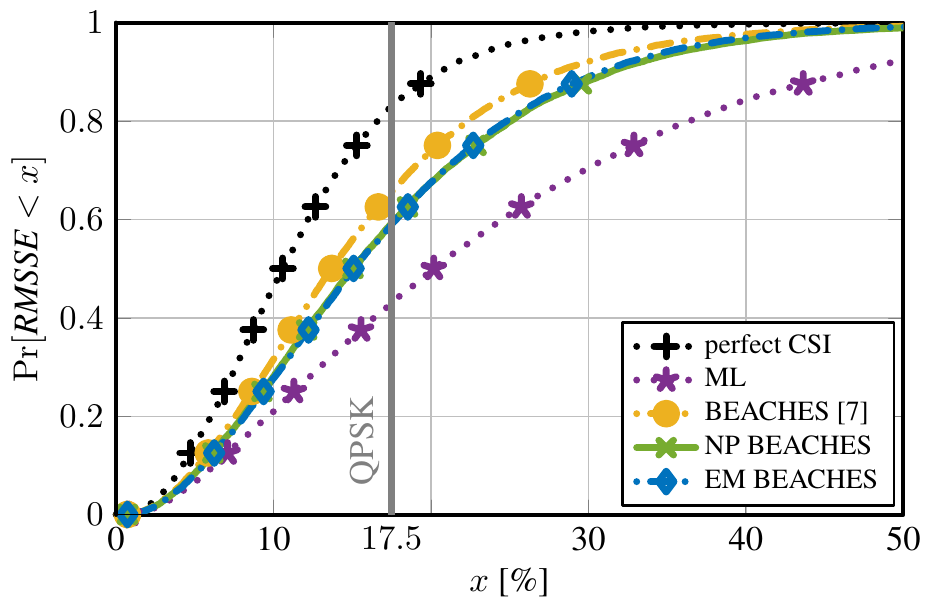}}
\caption{Cumulative distribution function (CDF) of the (a) mean-square error (MSE) and (b) root-mean-squared-symbol error (RMSSE) of channel estimation methods in a cell-free wireless system.}
\label{fig:cell_free}
\end{figure*}

 \section{Conclusions}
We have proposed blind estimators  for the average noise power, signal power, SNR, and MSE. Our estimators can be calculated at low complexity and only require the noisy observation vector, avoiding the need for additional pilot signals entirely. 
We have analyzed our estimators for a Bernoulli complex Gaussian sparsity model and evaluated their accuracy via simulations. 
Using three channel-vector denoising tasks in (i) a multi-antenna mmWave system, (ii) a 1-bit quantized multi-antenna mmWave system, and (iii) a cell-free system, we have demonstrated that our blind estimators can be used to develop a novel nonparametric denoiser that achieves comparable performance and the same complexity as BEACHES in \cite{ghods19a,mirfarshbafan2019beamspace} which requires knowledge of the average noise power. 
We believe that the proposed blind estimators find potential use in a large number of other wireless communication applications that contain sparse complex-valued signals.  

There are many avenues for future work. For signals that are less sparse (i.e., $p>0.421$), one may want to replace the median by a higher quantile and the scaling factor {$\log(2)$} needs to be adjusted accordingly---a derivation of such estimators would follow  immediately from our results in~\fref{sec:noise_estimator_analysis}.
Huber M-estimators \cite{huber64a} combine the idea of mean and median, and they may also prove useful for blind noise power estimation in the presence of sparse signals.
In the case of non-Gaussian, non-circularly-symmetric, or non-i.i.d sparse signals, new estimators can be tailored to exploit  specific statistical properties (e.g., structured sparsity). Extending the statistical model, e.g., to signals with correlation or structured sparsity, can lead to improved estimators and is left for future work.  In the case of colored noise (e.g., stemming from interference or large variations in radio-frequency circuitry), noise whitening techniques could be considered.

\appendices
\section{Proof of \fref{thm:mainresult}}
\label{app:mainresult}
 
\subsection{Prerequisites}
In what follows, we will need the distribution of $\bmz \define \abssquared{\bmy}$, where we assume $\bmy$ is distributed according to \fref{def:noisyBCG}.
Given a circularly-symmetric complex Gaussian RV $A$ with variance $E_a$, the RV $B=|A|^2$ is exponentially distributed with CDF $F_B(b) \define 1 - e^{-\frac{b}{{E_a}}}$, $b\geq0$.
Then, the CDF of each entry of the absolute-square noisy observation is as follows.

\begin{definition}[Noisy BCG Power RV] \label{def:noisyBCGsquared}
Let $\bmy$ be as in \fref{def:noisyBCG} and let $\bmz \define \abssquared{\bmy}$.
Then, for $z_d\geq0$, the CDF of each  entry of $\bmz$ is given by
\begin{align} \label{eq:noisyBCGpower}
F_Z(z_d) \define \conditionaltextstyle (1-p)\left(1 - e^{-\frac{z_d}{\No}}\right)  
 + p\left(1 - e^{-\frac{z_d}{\No+\Eh}}\right)\!.
\end{align}
\end{definition}

\subsection{Upper Bounds on the Median}
We start with the following two upper bounds on the median $\median_Z$ of a noisy BCG power RV~$Z$ with CDF given in \fref{eq:noisyBCGpower}. 

\begin{lemma}  \label{lem:smollemma1}
For a noisy BCG power RV in \fref{def:noisyBCGsquared} with $p < 0.5$, the median is bounded  from above by
\begin{align}
\median_Z \leq \conditionaltextstyle \No\log\!\left(\frac{2-2p}{1-2p}\right).
\end{align}
\end{lemma}

\begin{proof}
We start from the definition of the median in \fref{eq:median} for the RV $Z$ with CDF as in~\fref{eq:noisyBCGpower}:
\begin{align}
(1-p)\left(1 - e^{-\frac{\median_Z}{\No}}\right) + p\left(1 - e^{-\frac{\median_Z}{\No+\Eh}}\right) & =\conditionaltextstyle  \frac{1}{2}. \label{eq:niceform}
\end{align}
Since the second term is nonnegative, we can omit it to obtain the following inequality:
\begin{align} \label{eq:expression0}
(1-p)\left(1 - e^{-\frac{\median_Z}{\No}}\right)   & \leq\conditionaltextstyle \frac{1}{2}.
\end{align}
Note that this bound will be useful for vectors $\bms$ that are sparse, i.e., where $p$ is small. We can simplify \fref{eq:expression0} as 
\begin{align}
1 - e^{-\frac{\median_Z}{\No}}   & \leq\conditionaltextstyle \frac{1}{2(1-p)} \\
\conditionaltextstyle \conditionaltextstyle \log\left(\frac{1-2p}{2-2p}\right) & \conditionaltextstyle \leq -\frac{\median_Z}{\No} \label{eq:logarithmcondition}, 
\end{align}
which leads to the upper bound on the median $\median_Z$.
In order to take the logarithm in \fref{eq:logarithmcondition}, we require $p\in(0,0.5)$.
\end{proof}

\begin{lemma} \label{lem:smollemma2}
For a noisy BCG power RV $Z$ in \fref{def:noisyBCGsquared} with $p\leq \frac{1/2 - e^{-2} }{ 1 - e^{-2} } \approx 0.421$, the median is bounded  from above by
\begin{align}
\median_Z \leq \log(2)(\No + \Es).
\end{align}
\end{lemma}

\begin{proof}
We start from the definition of the median as in \fref{eq:niceform}.
Let us define the function $g(r) \define e^{-1/r}$ with $r>0$. We can now rewrite \fref{eq:niceform} as follows:
\begin{align}
\conditionaltextstyle \frac{1}{2} &= \conditionaltextstyle (1-p)g\!\left(\frac{\No}{\median_Z}\right)  + p g\!\left(\frac{\No+\Eh}{\median_Z}\right)\!. \label{eq:nicejensenexpression}
\end{align}
The function $g(r)$ is concave for $r \geq 1/2$. Therefore, to ensure concavity of $g(r)$ in \fref{eq:nicejensenexpression}, we~need
\begin{align}
\conditionaltextstyle \frac{\No}{\median_Z}  \geq \frac{1}{2} \quad \text{and} \quad 
\frac{\No+\Eh}{\median_Z}  \geq \frac{1}{2}. \label{eq:twoconcavityconditions}
\end{align}
The two conditions in \fref{eq:twoconcavityconditions} are guaranteed as long as $\protect{2\No\geq\median_Z}$. Because CDFs are nondecreasing functions, requiring $\protect{2\No\geq\median_Z}$ is equivalent to requiring $\protect{F_Z(2\No)\geq F_Z(\median_Z)=1/2}$, {which we can simplify as}
{\begin{align}
\underbrace{(1-p)\left(1 - e^{-\frac{2\No}{\No}}\right) + p\left(1 - e^{-\frac{2\No}{\No+\Eh}}\right)}_{F_Z(2\No)}  \geq \underbrace{\conditionaltextstyle\frac{1}{2}}_ {F_Z(\median_Z)} \\
\conditionaltextstyle \frac{1}{2} - e^{-2}  \conditionaltextstyle \geq   p \left(  e^{-\frac{2}{1+\Es/(p\No)}} - e^{-2} \right). \label{eq:pconditionSNR}
\end{align}}
Finally, to ensure \fref{eq:pconditionSNR} holds for all values of $\Es$ and $\No$, we require
\begin{align}
\conditionaltextstyle\frac{1}{2} - e^{-2} & \geq   p \left(  1 - e^{-2} \right)\!,
\end{align}
which implies that the condition $p\leq\pmax$ in \fref{eq:probabilitycondition} ensures concavity of $g(r)$. 
Then, assuming $p\leq\pmax$, we can now use Jensen's inequality on the expression in \fref{eq:nicejensenexpression} to get
\begin{align}
\conditionaltextstyle\frac{1}{2} \!\leq\!  \conditionaltextstyle g\left((1\!-\!p)\frac{\No}{\median_Z} + p\frac{\No+\Eh}{\median_Z}\right) \!=\! g\left(\frac{\No+\Es}{\median_Z}\right)\!.
\end{align}
We can now simplify this expression to
\begin{align}
\median_Z & \leq \log(2)(\No + \Es), \label{eq:probabilityupperbound}
\end{align}
which is the inequality in \fref{lem:smollemma2}.
\end{proof} 

\subsection{Lower Bound on the Median}
We now establish the following lower bound on the median. 
\begin{lemma} \label{lem:smollemma3}
For a noisy BCG power RV $Z$ in \fref{def:noisyBCGsquared} with $p \in (0,1]$, the median is bounded  from below by
\begin{align} \label{eq:bound3}
\conditionaltextstyle \frac{\log(2)\No}{ (1-p)+\frac{p^2}{p+\SNR}  } &  \leq \median_Z.
\end{align}
\end{lemma}

\begin{proof}
We start from the definition of the median as in \fref{eq:niceform}.
Since the exponential CDF $F_B(b) \define 1 - e^{-\frac{b}{{E_a}}}$ for ${E_a}\geq0$ is concave in $b$, Jensen's inequality leads to
\begin{align}
1 - e^{-(1-p)\frac{\median_Z}{\No}-p\frac{\median_Z}{\No+\Eh}} & \conditionaltextstyle \geq \frac{1}{2}.
\end{align}
We can simplify this expression  to obtain the following bound
\begin{align}
\conditionaltextstyle \frac{1}{2} & \conditionaltextstyle\geq  e^{-(1-p)\frac{\median_Z}{\No}-p\frac{\median_Z}{\No+\Eh}} \\
\log(1/2) &\conditionaltextstyle \geq -(1-p)\frac{\median_Z}{\No}-p\frac{\median_Z}{\No+\Eh},
\end{align}
which leads to the inequality in \fref{lem:smollemma3} we wanted to prove.
\end{proof}

\subsection{Combining the Results}
For all values of $p\in(0,1]$ and $\SNR\geq0$, we have that
\begin{align}
\conditionaltextstyle \frac{\median_Z}{\log(2)}\left( (1-p)+\frac{p^2}{p+\SNR} \right) \leq \frac{\median_Z}{\log(2)}, \label{eq:uperuperbound}
 \end{align}
and we defined $\estimatedNo$ such that $\plimNo {\median_Z}/{\log(2)}$ according to \fref{lem:convergenceofmedian}.

Finally, we can combine \fref{eq:uperuperbound} with \fref{lem:smollemma1}, \fref{lem:smollemma2} and \fref{lem:smollemma3} to obtain~\fref{eq:yummysandwichbound}.

\section{Proof of \fref{cor:errorbound}}
\label{app:errorproof}

\begin{proof}
Let the relative error of \fref{est:noisevariance} be $\protect{\varepsilon \define {|\estimatedNo-\No|}/{\No}}$.
Using the inequalities from \fref{thm:mainresult} and the quantities $\LB$ and $\UB$ defined there, we can bound $\varepsilon$~as follows:
\begin{align} \label{eq:relativeerrorbound}
\conditionaltextstyle \frac{\displaystyle\estimatedNolargeD-\UB}{\UB} \leq \displaystyle \lim_{D\to\infty}\varepsilon \leq \conditionaltextstyle \frac{\displaystyle\estimatedNolargeD-\LB}{\LB}.
\end{align}
By using $\plimNo {\median_Z}/{\log(2)}$ and replacing $\LB$ from \fref{eq:LB} and $\UB$ from \fref{eq:UB} into \fref{eq:relativeerrorbound}, after some simplifications, we obtain \fref{eq:errorbound}.
\end{proof}

\balance


\end{document}